\documentclass[12pt]{amsart}
\usepackage{txfonts}      
\usepackage{amssymb}
\usepackage{eucal}
\usepackage{amsmath}
\usepackage{amscd}
\usepackage[dvips]{color}
\usepackage{multicol}
\usepackage[all]{xy}           
\usepackage{graphicx}
\usepackage{color}
\usepackage{colordvi}
\usepackage{xspace}
\usepackage[hypertex]{hyperref}

\usepackage[active]{srcltx} 

\begin{document}  

\newcommand{\nc}{\newcommand}
\newcommand{\delete}[1]{}

\nc{\mlabel}[1]{\label{#1}}  
\nc{\mcite}[1]{\cite{#1}}  
\nc{\mref}[1]{\ref{#1}}  
\nc{\mbibitem}[1]{\bibitem{#1}} 

\delete{
\nc{\mlabel}[1]{\label{#1}  
{\hfill \hspace{1cm}{\bf{{\ }\hfill(#1)}}}}
\nc{\mcite}[1]{\cite{#1}{{\bf{{\ }(#1)}}}}  
\nc{\mref}[1]{\ref{#1}{{\bf{{\ }(#1)}}}}  
\nc{\mbibitem}[1]{\bibitem[\bf #1]{#1}} 
}

\newtheorem{theorem}{Theorem}[section]
\newtheorem{prop}[theorem]{Proposition}
\newtheorem{defn}[theorem]{Definition}
\newtheorem{lemma}[theorem]{Lemma}
\newtheorem{coro}[theorem]{Corollary}
\newtheorem{prop-def}{Proposition-Definition}[section]
\newtheorem{claim}{Claim}[section]
\newtheorem{remark}[theorem]{Remark}
\newtheorem{propprop}{Proposed Proposition}[section]
\newtheorem{conjecture}{Conjecture}
\newtheorem{exam}[theorem]{Example}
\newtheorem{assumption}{Assumption}
\newtheorem{condition}[theorem]{Assumption}
\newtheorem{notation}[theorem]{Notation}

\renewcommand{\labelenumi}{{\rm(\alph{enumi})}}
\renewcommand{\theenumi}{\alph{enumi}}

\nc{\tred}[1]{\textcolor{red}{#1}}
\nc{\tblue}[1]{\textcolor{blue}{#1}}
\nc{\tgreen}[1]{\textcolor{green}{#1}}
\nc{\tpurple}[1]{\textcolor{purple}{#1}}
\nc{\btred}[1]{\textcolor{red}{\bf #1}}
\nc{\btblue}[1]{\textcolor{blue}{\bf #1}}
\nc{\btgreen}[1]{\textcolor{green}{\bf #1}}
\nc{\btpurple}[1]{\textcolor{purple}{\bf #1}}

\nc{\li}[1]{\textcolor{red}{Li:#1}}
\nc{\cm}[1]{\textcolor{blue}{Chengming: #1}}
\nc{\xiang}[1]{\textcolor{red}{{\it Xiang: #1}}}

\nc{\dual}{\ast}
\nc{\otr}[1]{\bar{#1}}
\nc{\ad}{\mathrm{ad}}
\nc{\gl}{\mathfrak{gl}}
\nc{\ttl}[1]{\hat{#1}}
\nc{\TTL}{\wedge}
\nc{\ltt}[1]{\check{#1}}
\nc{\LTT}{\vee}
\nc{\ttt}[1]{\tilde{#1}}
\nc{\TTT}{T}
\nc{\adec}{\check{;}} \nc{\aop}{\alpha}
\nc{\dftimes}{\widetilde{\otimes}} \nc{\dfl}{\succ} \nc{\dfr}{\prec}
\nc{\dfc}{\circ} \nc{\dfb}{\bullet} \nc{\dft}{\star}
\nc{\dfcf}{{\mathbf k}} \nc{\apr}{\ast} \nc{\spr}{\cdot}
\nc{\twopr}{\circ} \nc{\sempr}{\ast} \nc{\bwt}{{mass}\xspace}
\nc{\bwts}{{masses}\xspace} \nc{\bop}{{extention}\xspace}
\nc{\ewt}{{mass}\xspace} \nc{\ewts}{{masses}\xspace}
\nc{\disp}[1]{\displaystyle{#1}} \nc{\tto}{{extended}\xspace}
\nc{\Tto}{{Extended}\xspace} \nc{\tte}{{extended}\xspace}
\nc{\gyb}{{generalized}\xspace} \nc{\Gyb}{{Generalized}\xspace}
\nc{\ECYBE}{{ECYBE}\xspace} \nc{\GAYBE}{{GAYBE}\xspace}
\nc{\mltimes}{\ltimes}

\nc{\bin}[2]{ (_{\stackrel{\scs{#1}}{\scs{#2}}})}  
\nc{\binc}[2]{ \left (\!\! \begin{array}{c} \scs{#1}\\
    \scs{#2} \end{array}\!\! \right )}  
\nc{\bincc}[2]{  \left ( {\scs{#1} \atop
    \vspace{-.5cm}\scs{#2}} \right )}  
\nc{\sarray}[2]{\begin{array}{c}#1 \vspace{.1cm}\\ \hline
    \vspace{-.35cm} \\ #2 \end{array}}
\nc{\bs}{\bar{S}} \nc{\dcup}{\stackrel{\bullet}{\cup}}
\nc{\dbigcup}{\stackrel{\bullet}{\bigcup}} \nc{\etree}{\big |}
\nc{\la}{\longrightarrow} \nc{\fe}{\'{e}} \nc{\rar}{\rightarrow}
\nc{\dar}{\downarrow} \nc{\dap}[1]{\downarrow
\rlap{$\scriptstyle{#1}$}} \nc{\uap}[1]{\uparrow
\rlap{$\scriptstyle{#1}$}} \nc{\defeq}{\stackrel{\rm def}{=}}
\nc{\dis}[1]{\displaystyle{#1}} \nc{\dotcup}{\,
\displaystyle{\bigcup^\bullet}\ } \nc{\sdotcup}{\tiny{
\displaystyle{\bigcup^\bullet}\ }} \nc{\hcm}{\ \hat{,}\ }
\nc{\hcirc}{\hat{\circ}} \nc{\hts}{\hat{\shpr}}
\nc{\lts}{\stackrel{\leftarrow}{\shpr}}
\nc{\rts}{\stackrel{\rightarrow}{\shpr}} \nc{\lleft}{[}
\nc{\lright}{]} \nc{\uni}[1]{\tilde{#1}} \nc{\wor}[1]{\check{#1}}
\nc{\free}[1]{\bar{#1}} \nc{\den}[1]{\check{#1}} \nc{\lrpa}{\wr}
\nc{\curlyl}{\left \{ \begin{array}{c} {} \\ {} \end{array}
    \right .  \!\!\!\!\!\!\!}
\nc{\curlyr}{ \!\!\!\!\!\!\!
    \left . \begin{array}{c} {} \\ {} \end{array}
    \right \} }
\nc{\leaf}{\ell}       
\nc{\longmid}{\left | \begin{array}{c} {} \\ {} \end{array}
    \right . \!\!\!\!\!\!\!}
\nc{\ot}{\otimes} \nc{\sot}{{\scriptstyle{\ot}}}
\nc{\otm}{\overline{\ot}} \nc{\ora}[1]{\stackrel{#1}{\rar}}
\nc{\ola}[1]{\stackrel{#1}{\la}}
\nc{\pltree}{\calt^\pl} \nc{\epltree}{\calt^{\pl,\NC}}
\nc{\rbpltree}{\calt^r} \nc{\scs}[1]{\scriptstyle{#1}}
\nc{\mrm}[1]{{\rm #1}}
\nc{\dirlim}{\displaystyle{\lim_{\longrightarrow}}\,}
\nc{\invlim}{\displaystyle{\lim_{\longleftarrow}}\,}
\nc{\mvp}{\vspace{0.5cm}} \nc{\svp}{\vspace{2cm}}
\nc{\vp}{\vspace{8cm}} \nc{\proofbegin}{\noindent{\bf Proof: }}
\nc{\proofend}{$\blacksquare$ \vspace{0.5cm}}
\nc{\freerbpl}{{F^{\mathrm RBPL}}}
\nc{\sha}{{\mbox{\cyr X}}}  
\nc{\ncsha}{{\mbox{\cyr X}^{\mathrm NC}}} \nc{\ncshao}{{\mbox{\cyr
X}^{\mathrm NC,\,0}}}
\nc{\shpr}{\diamond}    
\nc{\shprm}{\overline{\diamond}}    
\nc{\shpro}{\diamond^0}    
\nc{\shprr}{\diamond^r}     
\nc{\shpra}{\overline{\diamond}^r} \nc{\shpru}{\check{\diamond}}
\nc{\catpr}{\diamond_l} \nc{\rcatpr}{\diamond_r}
\nc{\lapr}{\diamond_a} \nc{\sqcupm}{\ot} \nc{\lepr}{\diamond_e}
\nc{\vep}{\varepsilon} \nc{\labs}{\mid\!} \nc{\rabs}{\!\mid}
\nc{\hsha}{\widehat{\sha}} \nc{\lsha}{\stackrel{\leftarrow}{\sha}}
\nc{\rsha}{\stackrel{\rightarrow}{\sha}} \nc{\lc}{\lfloor}
\nc{\rc}{\rfloor} \nc{\tpr}{\sqcup} \nc{\nctpr}{\vee}
\nc{\plpr}{\star} \nc{\rbplpr}{\bar{\plpr}} \nc{\sqmon}[1]{\langle
#1\rangle} \nc{\forest}{\calf} \nc{\ass}[1]{\alpha({#1})}
\nc{\altx}{\Lambda_X} \nc{\vecT}{\vec{T}} \nc{\onetree}{\bullet}
\nc{\Ao}{\check{A}} \nc{\seta}{\underline{\Ao}}
\nc{\deltaa}{\overline{\delta}} \nc{\trho}{\tilde{\rho}}

\nc{\mmbox}[1]{\mbox{\ #1\ }} \nc{\ann}{\mrm{ann}}
\nc{\Aut}{\mrm{Aut}} \nc{\can}{\mrm{can}} \nc{\twoalg}{{two-sided
algebra}\xspace} \nc{\colim}{\mrm{colim}} \nc{\Cont}{\mrm{Cont}}
\nc{\rchar}{\mrm{char}} \nc{\cok}{\mrm{coker}} \nc{\dtf}{{R-{\rm
tf}}} \nc{\dtor}{{R-{\rm tor}}}
\renewcommand{\det}{\mrm{det}}
\nc{\depth}{{\mrm d}} \nc{\Div}{{\mrm Div}} \nc{\End}{\mrm{End}}
\nc{\Ext}{\mrm{Ext}} \nc{\Fil}{\mrm{Fil}} \nc{\Frob}{\mrm{Frob}}
\nc{\Gal}{\mrm{Gal}} \nc{\GL}{\mrm{GL}} \nc{\Hom}{\mrm{Hom}}
\nc{\hsr}{\mrm{H}} \nc{\hpol}{\mrm{HP}} \nc{\id}{\mrm{id}}
\nc{\im}{\mrm{im}} \nc{\incl}{\mrm{incl}} \nc{\length}{\mrm{length}}
\nc{\LR}{\mrm{LR}} \nc{\mchar}{\rm char} \nc{\NC}{\mrm{NC}}
\nc{\mpart}{\mrm{part}} \nc{\pl}{\mrm{PL}} \nc{\ql}{{\QQ_\ell}}
\nc{\qp}{{\QQ_p}} \nc{\rank}{\mrm{rank}} \nc{\rba}{\rm{RBA }}
\nc{\rbas}{\rm{RBAs }} \nc{\rbpl}{\mrm{RBPL}} \nc{\rbw}{\rm{RBW }}
\nc{\rbws}{\rm{RBWs }} \nc{\rcot}{\mrm{cot}}
\nc{\rest}{\rm{controlled}\xspace} \nc{\rdef}{\mrm{def}}
\nc{\rdiv}{{\rm div}} \nc{\rtf}{{\rm tf}} \nc{\rtor}{{\rm tor}}
\nc{\res}{\mrm{res}} \nc{\SL}{\mrm{SL}} \nc{\Spec}{\mrm{Spec}}
\nc{\tor}{\mrm{tor}} \nc{\Tr}{\mrm{Tr}} \nc{\mtr}{\mrm{sk}}

\nc{\ab}{\mathbf{Ab}} \nc{\Alg}{\mathbf{Alg}}
\nc{\Algo}{\mathbf{Alg}^0} \nc{\Bax}{\mathbf{Bax}}
\nc{\Baxo}{\mathbf{Bax}^0} \nc{\RB}{\mathbf{RB}}
\nc{\RBo}{\mathbf{RB}^0} \nc{\BRB}{\mathbf{RB}}
\nc{\Dend}{\mathbf{DD}} \nc{\bfk}{{\bf k}} \nc{\bfone}{{\bf 1}}
\nc{\base}[1]{{a_{#1}}} \nc{\detail}{\marginpar{\bf More detail}
    \noindent{\bf Need more detail!}
    \svp}
\nc{\Diff}{\mathbf{Diff}} \nc{\gap}{\marginpar{\bf
Incomplete}\noindent{\bf Incomplete!!}
    \svp}
\nc{\FMod}{\mathbf{FMod}} \nc{\mset}{\mathbf{MSet}}
\nc{\rb}{\mathrm{RB}} \nc{\Int}{\mathbf{Int}}
\nc{\Mon}{\mathbf{Mon}}
\nc{\remarks}{\noindent{\bf Remarks: }}
\nc{\OS}{\mathbf{OS}} 
\nc{\Rep}{\mathbf{Rep}} \nc{\Rings}{\mathbf{Rings}}
\nc{\Sets}{\mathbf{Sets}} \nc{\DT}{\mathbf{DT}}

\nc{\BA}{{\mathbb A}} \nc{\CC}{{\mathbb C}} \nc{\DD}{{\mathbb D}}
\nc{\EE}{{\mathbb E}} \nc{\FF}{{\mathbb F}} \nc{\GG}{{\mathbb G}}
\nc{\HH}{{\mathbb H}} \nc{\LL}{{\mathbb L}} \nc{\NN}{{\mathbb N}}
\nc{\QQ}{{\mathbb Q}} \nc{\RR}{{\mathbb R}} \nc{\TT}{{\mathbb T}}
\nc{\VV}{{\mathbb V}} \nc{\ZZ}{{\mathbb Z}}


\nc{\calao}{{\mathcal A}} \nc{\cala}{{\mathcal A}}
\nc{\calc}{{\mathcal C}} \nc{\cald}{{\mathcal D}}
\nc{\cale}{{\mathcal E}} \nc{\calf}{{\mathcal F}}
\nc{\calfr}{{{\mathcal F}^{\,r}}} \nc{\calfo}{{\mathcal F}^0}
\nc{\calfro}{{\mathcal F}^{\,r,0}} \nc{\oF}{\overline{F}}
\nc{\calg}{{\mathcal G}} \nc{\calh}{{\mathcal H}}
\nc{\cali}{{\mathcal I}} \nc{\calj}{{\mathcal J}}
\nc{\call}{{\mathcal L}} \nc{\calm}{{\mathcal M}}
\nc{\caln}{{\mathcal N}} \nc{\calo}{{\mathcal O}}
\nc{\calp}{{\mathcal P}} \nc{\calr}{{\mathcal R}}
\nc{\calt}{{\mathcal T}} \nc{\caltr}{{\mathcal T}^{\,r}}
\nc{\calu}{{\mathcal U}} \nc{\calv}{{\mathcal V}}
\nc{\calw}{{\mathcal W}} \nc{\calx}{{\mathcal X}}
\nc{\CA}{\mathcal{A}}

\nc{\fraka}{{\mathfrak a}} \nc{\frakB}{{\mathfrak B}}
\nc{\frakb}{{\mathfrak b}} \nc{\frakd}{{\mathfrak d}}
\nc{\oD}{\overline{D}} \nc{\frakF}{{\mathfrak F}}
\nc{\frakg}{{\mathfrak g}} \nc{\frakm}{{\mathfrak m}}
\nc{\frakM}{{\mathfrak M}} \nc{\frakMo}{{\mathfrak M}^0}
\nc{\frakp}{{\mathfrak p}} \nc{\frakS}{{\mathfrak S}}
\nc{\frakSo}{{\mathfrak S}^0} \nc{\fraks}{{\mathfrak s}}
\nc{\os}{\overline{\fraks}} \nc{\frakT}{{\mathfrak T}}
\nc{\oT}{\overline{T}}
\nc{\frakX}{{\mathfrak X}} \nc{\frakXo}{{\mathfrak X}^0}
\nc{\frakx}{{\mathbf x}}
\nc{\frakTx}{\frakT}      
\nc{\frakTa}{\frakT^a}        
\nc{\frakTxo}{\frakTx^0}   
\nc{\caltao}{\calt^{a,0}}   
\nc{\ox}{\overline{\frakx}} \nc{\fraky}{{\mathfrak y}}
\nc{\frakz}{{\mathfrak z}} \nc{\oX}{\overline{X}}

\font\cyr=wncyr10

\nc{\redtext}[1]{\textcolor{red}{#1}}


\title{Generalizations of the classical Yang-Baxter equation and $\calo$-operators}

\author{Chengming Bai}
\address{Chern Institute of Mathematics\& LPMC, Nankai University, Tianjin 300071, P.R. China}
         \email{baicm@nankai.edu.cn}
\author{Li Guo}\thanks{Corresponding author: Li Guo, Department of Mathematics and Computer Science,
         Rutgers University,
         Newark, NJ 07102, U.S.A., E-mail: liguo@rutgers.edu}
\address{Department of Mathematics and Computer Science,
         Rutgers University,
         Newark, NJ 07102, U.S.A.}
\email{liguo@rutgers.edu}
\author{Xiang Ni}
\address{Chern Institute of Mathematics \& LPMC, Nankai
University, Tianjin 300071, P.R.
China}\email{xiangn$_-$math@yahoo.cn}



\begin{abstract}
Tensor solutions ($r$-matrices) of the classical Yang-Baxter equation (CYBE) in a Lie algebra, obtained as the  classical limit of the $R$-matrix
solution of the quantum Yang-Baxter equation (QYBE), is an important
structure appearing in different areas such as integrable systems,
symplectic geometry, quantum groups and quantum field theory.
Further study of CYBE led to its interpretation as certain operators, giving rise to the concept of
${\calo}$-operators. In~\mcite{Bai}, the
$\calo$-operators were in turn interpreted as tensor solutions of CYBE by enlarging the Lie algebra. The
purpose of this paper is to extend this study to a more
general class of operators that were recently
introduced~\mcite{BGN1} in the study of Lax pairs in integrable
systems. Relationship between $\calo$-operators, relative
differential operators and
Rota-Baxter operators are also discussed.
\end{abstract}

\subjclass[2000]{}

\keywords{Lie algebra, Yang-Baxter equation, $\calo$-operator, Rota-Baxter operator}

\maketitle

\tableofcontents

\setcounter{section}{0}

\section{Introduction}

This paper studies the relationship of a generalization of
$\calo$-operators with the classical Yang-Baxter equation (CYBE)~\mcite{BD} and its generalizations.

The CYBE in its original tensor form is the classical limit of the
quantum Yang-Baxter equation~\mcite{Ba,Ya} and has
played an important role in integrable systems~\mcite{BV1, BV2} and
Poisson-Lie groups (see \mcite{BGN1} and the
references therein). However the operator form of CYBE is often more
useful~\mcite{Se}. For example, the modified classical Yang-Baxter equation is
given in terms of the operator form~\mcite{Se,Bo}. This point of
view also allowed  Kupershmidt~\mcite{Ku} to generalize the notion
of (operator form of) CYBE to so-called $\calo$-operators, which
in fact can be traced back to Bordemann~\mcite{Bo} in integrable
systems.

\smallskip

It was shown in~\mcite{Bai} that an $\calo$-operator on
a Lie algebra can be realized as a tensor form solution of CYBE in a
larger Lie algebra. Thus these two seemingly distinct approaches to
solutions of the classical Yang-Baxter equation are unified.

Since then, both the tensor form approach and the operator form
approach of CYBE have been generalized. On one hand, the tensor form of CYBE has been generalized to
extended CYBE (ECYBE) and generalized CYBE (GCYBE) with the
latter arising naturally from the study of Lie
bialgebras~\mcite{Vai}.
On the other hand, the operator form of CYBE, in its generalized
form of $\calo$-operators, has been further generalized to
extended $\calo$-operators with modifications by several
parameters. These generalizations have found fruitful applications
to Lax pairs, Lie bialgebras and PostLie algebras~\mcite{BGN1}.
These generalizations have also motivated the study of their analogues for associative algebras~\mcite{BGN2}.

The purpose of this paper is to  further unify these
generalizations of the tensor forms and the operator
forms of CYBE in a similar framework as in~\mcite{Bai}. We also
study the relationship between Rota-Baxter operators and
$\calo$-operators, and study their differential analogues.

\smallskip

In Section~\mref{sec:ecy}, we study the relationship between
extended $\calo$-operators and extended CYBE. We then establish the
relationship between \tto $\calo$-operators and generalized CYBE in
Section~\mref{sec:gcyb}. Finally in Section~\mref{sec:rbod}, we
introduce a differential variation of an $\calo$-operator,
 called a relative differential operator. We then show that both an $\calo$-operator and
 a relative differential operator can be regarded as a Rota-Baxter operator on a larger Lie algebra.

\smallskip

\noindent {\bf Acknowledgements.} C. Bai thanks NSFC (10921061),
NKBRPC (2006CB805905) and SRFDP (200800550015) for support. L. Guo
thanks NSF grant DMS 1001855 for support and thanks the Chern Institute of Mathematics for hospitality. The authors thank the anonymous referee for helpful suggestions.

\section{Extended $\calo$-operators and ECYBE}
\mlabel{sec:ecy} We first recall in Section~\mref{ss:cyb} the
definitions of \tto $\calo$-operators, ECYBE and GCYBE. As a
motivation for our study, we also recall their relationship~\mcite{Bai} in the special case of $\calo$-operators
and CYBE. We then establish the relationship between \tto
$\calo$-operators and ECYBE in Section~\mref{ss:ecyb}.

\subsection{$\calo$-operators and CYBE}
\mlabel{ss:cyb} We first recall the classical result that a
skew-symmetric solution of CYBE in a Lie algebra gives an
$\calo$-operator through a duality between tensor
product and linear maps. Not every $\calo$-operator comes from a
solution of CYBE in this way. However, any $\calo$-operator can be
recovered from a solution of CYBE in a larger Lie algebra.

For the rest of the paper, $\bfk$ denotes a field whose
characteristic is not 2, unless otherwise stated.
A Lie algebra is taken to be a Lie algebra over $\bfk$. A tensor product is also taken over $\bfk$.

\subsubsection{From CYBE to $\calo$-operators}

Let $\frak{g}$ be a Lie algebra. For $r=\sum\limits_i a_i\ot b_i\in \frakg^{\ot 2}$,  we use the notations
(in the universal enveloping algebra $U(\frak g)$):
$$
r_{12}=\sum_ia_i\otimes b_i\otimes 1,\quad r_{13}=\sum_{i}a_i\otimes
1\otimes b_i,\quad r_{23}=\sum_i1\otimes a_i\otimes b_i,
$$
and
$$
[r_{12},r_{13}]=\sum_{i,j}[a_i,a_j]\otimes b_i\otimes b_j,\;
[r_{13},r_{23}]=\sum_{i,j}a_i\otimes
a_j\otimes[b_i,b_j],\;
[r_{12},r_{23}]=\sum_{i,j}a_i\otimes
[b_i,a_j]\otimes b_j\,.
$$

If $r\in \frakg^{\ot 2}$ satisfies the {\bf classical Yang-Baxter equation (CYBE)}
\begin{equation}
\textbf{C}(r)\equiv[r_{12},r_{13}]+[r_{12},r_{23}]+[r_{13},r_{23}]=0,
\mlabel{eq:cybe}
\end{equation}
then $r$ is called a {\bf solution of CYBE in $\frakg$}.

The {\bf twisting operator}
$\sigma: \frak{g}^{\ot 2} \rightarrow \frak{g}^{\ot 2}$ is defined by
$$
\sigma (x \otimes y) = y\otimes x,\quad \forall x, y\in \frak{g}.
$$
We call $r=\sum\limits_i a_i\ot b_i\in
\frak{g}^{\ot 2}$ {\bf skew-symmetric} (resp. {\bf symmetric}) if
$r=-\sigma(r)$ (resp. $r=\sigma(r)$).

We recall the following classical result~\mcite{Se} that establishes
the first connection between CYBE and certain linear operators that
had been called {\bf Rota-Baxter operators}~\mcite{Bax,R}
in the context of associative algebras which have found applications to renormalization in quantum field theory and number theory~\mcite{EGK,GZ} recently.
Let $\frakg$  be a Lie algebra with finite dimension over $\bfk$. Let
 $$\TTL:\frakg \ot \frakg \to \Hom(\frakg^*,\frakg) ,\quad  r\mapsto \ttl{r}, \quad \forall r\in \frakg\ot \frakg,$$
be the usual linear isomorphism, namely, for $r=\sum_i u_i\ot v_i\in \frakg^{\ot 2}$, we define
$$ \ttl{r}:\frakg^* \to \frakg, \quad \ttl{r}(a^*)=\sum_i a^*(u_i)v_i, \quad \forall a^*\in
\frakg^*.$$
 In other words,
$$ \langle \ttl{r}(a^*), b^*\rangle = \langle a^* \ot b^*, r\rangle, \quad a^*,b^*\in \frakg^*,$$
where, for a finite dimensional vector space,
$\langle\ , \ \rangle$ denotes the usual pairings $V\ot V^* \to \bfk$ and $V^*\ot V\to \bfk$.

Recall that a bilinear form $B(\ ,\ ):\frakg\ot \frakg \to \bfk$ is called {\bf invariant} if
$$B([x,y],z)=B(x,[y,z]), \quad \forall x, y, z\in
\frakg.$$

\begin{theorem} $($\mcite{Se}$)$
Suppose $\frakg$ has a nondegenerate and
symmetric bilinear form $B(\;,\;):\frakg\ot \frakg \to \bfk$ which is invariant, allowing us to identify $\frakg^*$ with $\frakg$.
 Let $r\in \frakg^{\ot 2}$ be skew-symmetric.
 Then $r$ is a solution  of CYBE if and only if $\ttl{r}:\frak g\rightarrow
 \frak g$ satisfies the {\bf Rota-Baxter equation} $($of weight 0$)$
\begin{equation}
[\ttl{r}(x),\ttl{r}(y)]=\ttl{r}([\ttl{r}(x),y]+[x,\ttl{r}(y)]), \quad \forall x,y\in \frakg.
\mlabel{eq:rb0}
\end{equation}
\mlabel{thm:sts}
\end{theorem}
Because of this theorem,
Eq.~(\mref{eq:rb0}) is called the {\bf operator form of CYBE} while
Eq.~(\mref{eq:cybe}) is called the {\bf tensor form of CYBE}.
Without assuming the
existence of a nondegenerate symmetric invariant bilinear form on
$\frakg$, it is known that the following result holds~(\mcite{YKS2}): if the symmetric part $r\in \frakg\ot
\frakg$ is invariant, then $r$ is a solution of the tensor form
of CYBE if and only if $\ttl{r}:\frakg^*\to \frakg$ satisfies
\begin{equation}
[r(a^*),r(b^*)]=r({\rm ad}^*(r(a^*))b^*-{\rm
ad}^*(r(b^*))a^*+[a^*,b^*]_{-}),\quad  \forall
a^*,b^*\in\frak{g}^*,\label{eq:cokuper}
\end{equation}
where $[,]_{-}$ is a Lie bracket on $\frak{g}^*$ defined by
\begin{equation}
[a^*,b^*]_{-}\equiv-{\rm ad}^*((r+r^t)(a^*))b^*,\quad \forall
a^*,b^*\in \frak{g}^*,\label{eq:pmbrac}
\end{equation}
with $r^t$ denoting the transpose of $r$. When $r$ is skew-symmetric, Eq.~(\mref{eq:cokuper}) becomes
\begin{equation}
[\ttl{r}(x),\ttl{r}(y)]
=\ttl{r}(\ad^*\ttl{r}(x)(y)-\ad^*\ttl{r}(y)(x)), \quad \forall
x,y\in \frakg^*. \mlabel{eq:kuper}
\end{equation}
which can be regarded as a generalization of the operator
form~(\mref{eq:rb0}) of CYBE.

There is a further generalization~\mcite{Bo,Ku} of Eq.~(\mref{eq:kuper}).

\begin{defn}
{\rm Let $\frakg$ be a Lie algebra. Let $V=(V,\rho)$ be a
$\frakg$-module, given by a representation $\rho:\frakg\to \frak
g\frak l(V)$ where $\frak g\frak l (V)$ is the Lie algebra on
$\End(V)$. Denote
$$g\cdot v:=\rho(g)(v), \quad \forall g\in \frakg,
v\in V.$$ A linear map $\alpha: V\to \frakg$ is called an {\bf
$\calo$-operator} if
\begin{equation}
[\alpha(x),\alpha(y)]=\alpha(
\alpha(x)\cdot y - \alpha(y)\cdot x), \quad \forall x,y\in V.
\mlabel{eq:oop}
\end{equation}
}
\end{defn}
Thus a skew-symmetric solution $r\in \frakg^{\ot 2}$ of CYBE gives an $\calo$-operator $\ttl{r}:\frakg^* \to \frakg$.

\subsubsection{From $\calo$-operators to CYBE}

In general it is not true that every $\calo$-operator $\alpha:\frakg^* \to \frakg$ comes from a skew-symmetric solution of CYBE in $\frakg$. As we will see next, such an $\alpha$ corresponds to a solution of CYBE in a larger Lie algebra.

The $\frakg$-module $(V,\rho)$ defines a Lie algebra  bracket $[\
,\ ]_\rho$ on $\frakg \oplus V$, called the {\bf semidirect
product} and denoted by $\frakg\ltimes_\rho V$, such
that
\begin{equation}
[g_1+v_1,g_2+v_2]_\rho=[g_1,g_2]+g_1\cdot v_2 -g_2\cdot v_1, \quad
\forall g_1,g_2\in \frakg, v_1, v_2 \in V. \mlabel{eq:sdir}
\end{equation}

Most of the following results are standard. But we want to use the notations throughout the rest of the paper.
\begin{prop}
Let $V$ and $W$ be finite dimensional vector spaces over $\bfk$.
\begin{enumerate}
\item
We have the natural isomorphisms
\begin{eqnarray}
&&\TTL:=\TTL_{V,W}: V\ot W \cong  V^{**} \ot W \cong \Hom(V^*,W),\quad r \mapsto \ttl{r}, \quad  \forall r\in V\ot W, \\
&&\LTT:=\LTT_{\Hom(V,W)}: \Hom(V,W) \cong V^* \ot W, \quad \alpha \mapsto \ltt{\alpha}, \quad \forall \alpha\in \Hom(V,W).
\end{eqnarray}
Thus the maps $\TTL_{V,W}$ and $\LTT_{\Hom(V^*,W)}$ are the inverses of each other.
\item
Define the {\bf twisting operator} by
\begin{equation}
\sigma: V\ot W \to W\ot V, \quad v\ot w \mapsto w\ot v, \quad \forall v\in V, w\in W.
\end{equation}
For $\alpha:V^*\to W$, let $\alpha^{\dual}:W^* \to V^{**}\cong V$ be the
dual map of $\alpha$.

Then for $r\in V\ot W$, we have
\begin{equation}
\widehat{\sigma(r)}=\ttl{r}^{\dual}.
\mlabel{eq:tt}
\end{equation}
\mlabel{it:dual}
\item
We also have the natural injections
\begin{eqnarray}
&\TTT:=\TTT_{V\ot W}:& V\ot W \to (V\oplus W)^{\ot 2}, \notag\\
&& \quad v\ot w \mapsto \widetilde{v\ot w}:=(v,0)\ot (0,w), \quad \forall v\in V, w\in W.
\mlabel{eq:TTTt}
\\
& \TTT: =\TTT_{\Hom(V,W)}:&
\Hom(V,W)\to \Hom(V\oplus W^*, V^*\oplus W), \notag \\
&& \alpha \mapsto \ttt{\alpha}:=\iota_2\circ \alpha \circ p_1, \quad \forall
\alpha \in \Hom(V,W).
\mlabel{eq:TTTl}
\end{eqnarray}

Here for vector spaces $V_i$, $i=1,2$, $\iota_i:V_i\to V_1\oplus V_2$ is the usual inclusion and $p_i: V_1\oplus V_2 \to V_i$ is the usual projection.
\item
We have the following commutative diagram.
\begin{equation}
\xymatrix{ \Hom(V,W) \ar^{\LTT}[rr] \ar_{\TTT}[d]
    & & V^* \ot W \ar^{\TTT}[d] \\
    \Hom(V\oplus W^*, V^*\oplus W) \ar^{\LTT}[rr] && (V^*\oplus W)^{\ot 2}
    }
\mlabel{eq:tttcom}
\end{equation}
\end{enumerate}
\mlabel{pp:tl}
\end{prop}
\begin{proof}
We just verify Item~(\mref{it:dual}). Let $\{e_1,...e_n\}$ be a basis of
$V$ and $\{f_1,...,f_m\}$ be a basis of $W$. Let
$\{e_1^*,...,e_n^*\}$ and $\{f_1^*,...,f_m^*\}$ be the corresponding
dual bases. Then for $\alpha=\sum_iv_i\otimes w_i\in V\otimes W$, where $v_i\in V,w_i\in W$, $\ttl{\alpha}$ is give by
$\ttl{\alpha}(e_k^*)=\sum_i\langle v_i,e_k^*\rangle w_i$, where
$1\leq k\leq n$. For any $1\leq s\leq m$, by definition we have
$$\sum_i\langle v_i,e_k^*\rangle\langle
w_i,f_s^*\rangle=\langle\ttl{\alpha}(e_k^*),f_s^*\rangle=\langle
e_k^*, (\ttl{\alpha})^{\dual}(f_s^*)\rangle.$$
Thus
$(\ttl{\alpha})^{\dual}(f_s^*)=\sum_iv_i\langle
w_i,f_s^*\rangle=\widehat{\sigma(\alpha)}(f_s^*)$, as required.
\end{proof}

Let $V$ be a vector space. We also use the following
notations:
\begin{equation}
r_\pm=(r\pm\sigma(r))/2, \quad \alpha_\pm:=(\alpha\pm \alpha^{\dual})/2,
\quad \forall r\in V^{\ot 2}, \alpha\in \Hom(V^*,V).
\mlabel{eq:alphabeta}
\end{equation}
Note that for any $r\in V\otimes V$, $(\ttl{r})_\pm =
\widehat{r_\pm}$ by Eq.~(\mref{eq:tt}). So the notation
$\ttl{r}_\pm$ is well-defined.

For a representation $\rho:\frak g\to \gl(V)$ of a Lie algebra $\frakg$, let $\rho^*:\frakg \to \gl(V^*)$ be the dual representation. Then
$\frakg\ltimes_{\rho^*} V^*$ is defined. Suppose that $\frakg$ and
$V$ are finite dimensional. Then using Proposition~\mref{pp:tl}, we
have the natural embedding
$$ \xymatrix{ \Hom(V,\frakg)\ \ \ar^{\LTT}@{>>->>}[r] & V^* \ot \frakg  \ \ \ar^{\TTT\ \ \ }@{>>->}[r] & (\frakg\ltimes_{\rho^*}V^*)^{\ot 2}, & \alpha \ar@{|->}[r] &
    \ltt{\alpha} \ar@{|->}[r] & \ttt{\ltt{\alpha}}}, \quad \alpha\in \Hom(V,\frakg).$$
The following result identifies any $\calo$-operator
as a solution of CYBE in a suitable Lie algebra.

\begin{theorem} (\mcite{Bai})
A linear map $\alpha: V\to \frakg$ is an $\calo$-operator if and only  if
$\ttt{\ltt{\alpha}}_- =(\ttt{\ltt{\alpha}}-\sigma(\ttt{\ltt{\alpha}}))/2$
is a skew-symmetric solution of CYBE in $\frakg\ltimes_{\rho^*}
V^*$. \mlabel{thm:bai}
\end{theorem}

\subsection{From extended CYBE to extended $\calo$-operators}
\mlabel{ss:ecyb}
Recently, the concepts of (the tensor form of) CYBE and $\calo$-operators have been generalized, and the connection from CYBE to $\calo$-operators has been generalized to this context.

\subsubsection{Extended CYBE}

For any $r=\sum\limits_ia_i\otimes b_i\in \frak
g\otimes \frak g$, we set
$$r_{21}=\sum_ib_i\otimes a_i\otimes 1,\quad r_{32}=\sum_i1\otimes
b_i\otimes a_i,\quad r_{31}=\sum_ib_i\otimes 1\otimes a_i.$$
Moreover, we set
\begin{equation}
 [(a_1\ot a_2\ot a_3), (b_1\ot b_2\ot b_3)] =
  [a_1, b_1]\ot [a_2, b_2] \ot [a_3, b_3], \quad \forall\, a_i,b_i\in \frakg, i=1,2,3.
\notag 
\end{equation}

\begin{defn}
{\rm Let $\frak{g}$ be a Lie algebra. Fix an $\epsilon\in
\bfk$. The equation
\begin{equation}
[r_{12},r_{13}]+[r_{12},r_{23}]+[r_{13},r_{23}]=\epsilon
[(r_{13}+r_{31}),(r_{23}+r_{32})] \mlabel{eq:ecybe}
\end{equation}
is called the {\bf \tte classical Yang-Baxter equation of \ewt
$\epsilon$} (or {\bf ECYBE of \ewt $\epsilon$} in short). }
\mlabel{de:ecybe}
\end{defn}

\subsubsection{Extended $\calo$-operators}

Our generalization of an $\calo$-operator was inspired by two
developments. On one hand, since an $\calo$-operator is a natural
generalization of a Rota-Baxter operator of weight 0, it is
desirable to define an $\calo$-operator of non-zero weight that
generalizes a Rota-Baxter of non-zero weight that was first defined
for associative algebras. On the other hand, Semenov-Tian-Shansky
\cite{Se} introduced the notion of {\bf modified Yang-Baxter
equation}
\begin{equation}
[R(x),R(y)]-R([R(x),y]+[x,R(y)])=-[x,y], \mlabel{eq:mbe}
\end{equation}
where $R:\frakg\rightarrow \frakg$ is a linear operator.  Moreover,
a {\bf Baxter Lie algebra} was introduced as a Lie algebra with a
linear operator $R$ satisfying the modified Yang-Baxter equation
\cite{BV1,Bo}.

These developments motivated us to give a framework of $\calo$-operators with extensions to uniformly treat these generalizations, and to generalize Theorem~\mref{thm:bai} to such extended $\calo$-operators.

We first introduce the basic Lie algebra setup.
\begin{defn}
 {\rm
\begin{enumerate}
\item Let $(\frakg,[\,,\,]_\frakg)$, or simply $\frakg$, denote a Lie algebra $\frakg$ with Lie bracket $[\,,\,]_\frakg$.
\item
For a Lie algebra $\frakb$, let ${\rm Der}_{\bfk}\frakb$ denote the Lie algebra of derivations of $\frakb$.
\item
Let $\fraka$ be a Lie algebra. An {\bf $\fraka$-Lie algebra} is a
triple $(\frakb,[\,,\,]_\frakb,\pi)$ consisting of a Lie algebra
$(\frakb,[\,,\,]_\frakb)$ and a Lie algebra homomorphism $\pi:\fraka
\to {\rm Der}_{\bfk}\frakb$. To
simplify the notation, we also let $(\frakb,\pi)$ or simply $\frakb$
denote $(\frakb,[\,,\,]_\frakb,\pi)$. \item
 Let $\fraka$ be a Lie algebra and let $(\frakg,\pi)$ be an $\fraka$-Lie algebra. Let $a\cdot b$ denote $\pi(a)b$ for $a\in \fraka$ and $b\in \frakg$.
\end{enumerate}
}
\mlabel{de:semidi}
\end{defn}

We recall a well-known result in Lie theory which will be
used throughout this paper.
\begin{prop}~\mcite{Kna}
 Let $\mathfrak{a}$ be a Lie algebra and let $(\mathfrak{b},\pi)$ be an $\fraka$-Lie algebra. Then there exists a unique Lie
algebra structure on the vector space direct sum
$\mathfrak{g}=\mathfrak{a}\oplus\mathfrak{b}$ retaining the old brackets in $\mathfrak{a}$ and $\mathfrak{b}$ and satisfying
$[x,y]=\pi(x)y$ for $x\in\mathfrak{a}$ and $y\in\mathfrak{b}$.
\mlabel{pp:semidi}
\end{prop}

\begin{remark}
 {\rm  The Lie algebra $\mathfrak{g}$ in the above
proposition is called the {\bf semidirect product} of $\mathfrak{a}$
and $\mathfrak{b}$, and we write it as
$\mathfrak{g}=\mathfrak{a}\mltimes_{\pi}\mathfrak{b}$. When the Lie
bracket in $\mathfrak{b}$ happens to be trivial, i.e., it is a
vector space, then we recover the usual semidirect product in Eq.~(\mref{eq:sdir}).}
\end{remark}

\begin{defn}
{\rm Let $\frak{g}$ be a Lie algebra.
\begin{enumerate}
\item
 Let $\kappa\in\bfk$ and
let $(V,\rho)$ be a $\frak{g}$-module. A linear map
$\beta:V\to\frak{g}$ is called an {\bf antisymmetric
$\frak{g}$-module homomorphism of \ewt $\kappa$} if
\begin{eqnarray}
& &\kappa\beta(x)\cdot y+\kappa\beta(y)\cdot x=0 ,\mlabel{eq:thantisymmetry}\\
& & \kappa\beta(\xi\cdot
x)=\kappa[\xi,\beta(x)]_{\frak{g}},\quad \forall x,y\in V, \xi\in\frak{g}.
\mlabel{eq:thginvariance}
\end{eqnarray}

\item
Let $\kappa,\mu\in\bfk$ and let $(\frak{k},\pi)$ be a $\frak{g}$-Lie
algebra. A linear map $\beta:\frak{k}\rightarrow \frak{g}$ is called an
{\bf antisymmetric $\frak{g}$-module homomorphism of \ewt
$(\kappa,\mu)$} if $\beta$ satisfies Eq.~(\mref{eq:thantisymmetry}),
Eq.~(\mref{eq:thginvariance}) and the following equation:
\begin{equation}
\mu\beta([x,y]_{\frak{k}})\cdot z=\mu[\beta(x)\cdot
y,z]_{\frak{k}},\quad \forall x,y,z\in\frak{k}. \mlabel{eq:thequivalence}
\end{equation}
 \end{enumerate}
}
\end{defn}
\begin{defn}
{\rm  Let $\frak{g}$ be a Lie algebra and let $(\frak{k},\pi)$ be a
$\frak{g}$-Lie algebra.
\begin{enumerate}
\item
Let $\lambda,\kappa,\mu\in\bfk$. Fix an antisymmetric
$\frak{g}$-module homomorphism $\beta:V\to \frak{g}$ of \ewt
$(\kappa,\mu)$. A linear map $\alpha:\frak{k}\to \frak{g}$ is called an
{\bf extended $\calo$-operator of weight $\lambda$ with extension
$\beta$ of \ewt $(\kappa,\mu)$} if:
\begin{equation}
[\alpha(x),\alpha(y)]_{\frak{g}}-\alpha(\alpha(x)\cdot y-\alpha(y)\cdot
x+\lambda[x,y]_{\frak{k}})=\kappa[\beta(x),\beta(y)]_{\frak{g}} +\mu\beta([x,y]_{\frak{k}}),\;\;\forall
x,y\in\frak{k}. \mlabel{eq:gmcybe}
\end{equation}
\item
We also let $(\alpha,\beta)$ denote an extended $\calo$-operator with
extension $\beta$.
\item
When $(V,\rho)$ is a $\frak{g}$-module, we regard $(V,\rho)$ as a
$\frak{g}$-Lie algebra with the trivial bracket. Then $\lambda,\mu$ are
irrelevant. We then call the pair {\bf $(\alpha,\beta)$ an extended
$\calo$-operator with extension $\beta$ of \ewt $\kappa$}.
\end{enumerate}
}
\end{defn}

\begin{defn}
{\rm Let $\frak{g}$ be a Lie algebra and $(\frak{k},\pi)$ be a
$\frak{g}$-Lie algebra. Then $\alpha:\frak{k}\to\frak{g}$ is
called an {\bf $\calo$-operator of weight $\lambda\in\bfk$} if it satisfies
\begin{equation}
[\alpha(x),\alpha(y)]_{\frak{g}}=\alpha\big(\alpha(x)\cdot y-\alpha(y)\cdot
x+\lambda[x,y]_{\frak{k}}\big),\;\;\forall x,y\in\frak{k}.
\mlabel{eq:gcybe}
\end{equation}
When $(\frak{k},\pi)=(\frak{g},{\rm ad})$, Eq.~(\mref{eq:gcybe})
takes the following form:
\begin{equation}
[\alpha(x),\alpha(y)]_{\frak{g}}=\alpha([\alpha(x),
y]_{\frak{g}}+[x,\alpha(y)]_{\frak{g}} +\lambda[x,y]_{\frak{g}}),\;\;\forall
x,y\in\frak{g}. \mlabel{eq:rotabaxter}
\end{equation}
A linear endomorphism $\alpha:\frak{g}\to\frak{g}$ satisfying
Eq.~(\mref{eq:rotabaxter}) is called a {\bf Rota-Baxter operator of
weight $\lambda$}~\mcite{Bax,EGK,GZ,R} (in the Lie algebra context).}
\end{defn}

\subsubsection{From extended CYBE to extended $\calo$-operators}
We next generalize Theorem~\mref{thm:sts}.
\begin{lemma}{\rm (\mcite{BGN1})}
Let $\frak{g}$ be a Lie algebra with finite $\bfk$-dimension
and $r\in \frak{g}\otimes\frak{g}$ be symmetric. Then the following
conditions are equivalent.
\begin{enumerate}
\item
 $r\in\frak{g}\otimes\frak{g}$ is {\bf invariant}, that is,
$({\rm ad}(x)\otimes \id+\id\otimes{\rm ad}(x))r=0, \forall x\in\frak{g}$;\mlabel{it:betainvariant}
\item
$\ttl{r}:\frak{g}^*\rightarrow\frak{g}$ is {\bf antisymmetric}, that
is, ${\rm ad}^*(\ttl{r}(a^*))b^*+{\rm ad}^*(\ttl{r}(b^*))a^*=0, \forall a^*,b^*\in\frak{g}^*$;\mlabel{it:betaantisymmetry}
\item
$\ttl{r}:\frak{g}^*\rightarrow\frak{g}$ is {\bf $\frakg$-invariant},
that is, $\ttl{r}({\rm ad}^*(x)a^*)=[x,\ttl{r}(a^*)], \forall x\in\frak{g}, a^*\in\frak{g}^*$.\mlabel{it:betaginvariant}
\end{enumerate} \mlabel{le:symmetry}
\end{lemma}

The following result characterizes solutions of ECYBE in a Lie algebra $\frakg$ in terms of \tto $\calo$-operators on $\frakg$.

\begin{theorem}~$($\mcite{BGN1}$)$
Let $\frakg$ be a Lie algebra with finite $\bfk$-dimension,
let $r\in \frakg\otimes \frakg$ and let $\ttl{r}: \frakg^*\to
\frakg$ be the corresponding linear map. Define $\ttl{r}_\pm$ by
Eq.~$($\mref{eq:alphabeta}$)$. Suppose that $r_+$ is invariant. Then
$r$ is a solution of ECYBE of \ewt $\frac{\kappa+1}{4}$:
\begin{equation}
[r_{12},r_{13}]+[r_{12},r_{23}]+[r_{13},r_{23}]=\frac{\kappa+1}{4}[(r_{13}+r_{31}),
(r_{23}+r_{32})] \notag
\end{equation}
if and only if $\ttl{r}_-$ is an \tto $\calo$-operator with \bop
$\ttl{r}_+$ of \bwt $\kappa$, i.e., the following equation holds:
\begin{equation}
[\ttl{r}_-(a^*),\ttl{r}_-(b^*)]-\ttl{r}_-({\rm ad}^*(\ttl{r}_-(a^*))b^*-{\rm
ad}^*(\ttl{r}_-(b^*))a^*)=\kappa[\ttl{r}_+(a^*),\ttl{r}_+(b^*)], \quad \forall
a^*,b^*\in\frak{g}^*.\label{eq:kmcybe}
\end{equation}
\mlabel{thm:cybea}
\end{theorem}

In the special case when $r_+=0$ (hence $\ttl{r}_+=0$), we obtain Kupershmidt's result Eq.~(\mref{eq:kuper})
and hence Theorem~\mref{thm:sts}.

\subsection{From \tto $\calo$-operators to ECYBE}
\mlabel{ss:tto}
We now start with an arbitrary \tto $\calo$-operator and characterize it as a solution of ECYBE in a suitable Lie algebra.

Let $\frak{g}$ be a Lie algebra and let $(V,\rho)$ be a
$\frak{g}$-module, both with finite $\bfk$-dimensions. Let
$(V^*,\rho^*)$ be the dual $\frak{g}$-module and let
$\ttt{\frak{g}}= \frak{g}\ltimes_{\rho^*}V^*$. Then from
Proposition~\mref{pp:tl}, we have the commutative diagram
\begin{equation}
\xymatrix{ \Hom(V,\frak{g}) \ar^{\LTT}[rr] \ar^{\TTT}[d] && \frak{g}\ot V^*  \ar^{\TTT}[d]\\
\Hom(\ttt{\frakg}^*,\ttt{\frakg}) \ar^{\LTT}[rr] & & \ttt{\frak{g}} \ot \ttt{\frak{g}} } \qquad
\xymatrix{ \beta \ar@{|->}[rr] \ar@{|->}[d] && \ltt{\beta} \ar@{|->}[d] \\ \ttt{\beta} \ar@{|->}[rr] && \ttt{\ltt{\beta}}=\ltt{\ttt{\beta}}
}
\mlabel{eq:ltl}
\end{equation}

\begin{lemma}
Let $\frak{g}$ be a Lie algebra and let $(V,\rho)$ be a
$\frak{g}$-module, both with finite $\bfk$-dimensions. Then
$\beta\in \Hom(V,\frak{g})$ is an antisymmetric $\frakg$-module
homomorphism of \ewt $\kappa$ if and only if $\ttt{\beta}_+ \in
\Hom({\ttt{\frak{g}}}^*,\ttt{\frak{g}})$ is an antisymmetric
$\ttt{\frak g}$-module homomorphism of \ewt $\kappa$.
\mlabel{le:syco}
\end{lemma}

\begin{proof}
The case that $\kappa=0$ is obvious. So we suppose that
$\kappa\ne0$. Then antisymmetric of \ewt $\kappa$ is the same as
antisymmetric (of \ewt 1) since we assume that $\bfk$ is a field.
Note that for any $a^*\in\frak{g}^*$ and $u\in V$, we have
$\ttt{\beta}_{+}(a^*)=\beta^*(a^*)/2$ and
$\tilde{\beta}_{+}(u)=\beta(u)/2$ where $\beta^*:\frak{g}^*\to V^*$
is the dual linear map associated to $\beta$. In fact, we have
$\tilde{\beta}_{+}=(\tilde{\beta}+{\tilde{\beta}}^{\dual})/2$. Moreover,
for any $a^*\in \frak g^*, u\in V$,
\begin{eqnarray*}
&&\tilde{\beta}(a^*) = l_2\circ\beta\circ p_1(a^*)=0,\\
&&\tilde{\beta}(u)= l_2\circ\beta\circ
p_1(u)=\beta(u).
\end{eqnarray*}
Hence for any $b^*\in \frak g^*, v\in V$,
{\allowdisplaybreaks
\begin{eqnarray*}
&&\langle {\tilde{\beta}}^{\dual}(a^*),v\rangle =\langle a^*,\tilde{\beta}(v)\rangle =\langle a^*,\beta(v)\rangle
=\langle \beta^*(a^*),v\rangle ; \\
&& \langle
\tilde{\beta}^{\dual}(a^*),b^*\rangle =\langle
a^*,\tilde{\beta}(b^*)\rangle =0;\\
&& \langle {\tilde{\beta}}^{\dual}(u),a^*\rangle =\langle u, \tilde{\beta}(a^*)\rangle =0;\\
&&\langle {\tilde{\beta}}^{\dual}(u),v\rangle =\langle
u,\tilde{\beta}(v)\rangle =\langle u,\beta(v)\rangle =0.
\end{eqnarray*}
}
So we have
 $\tilde{\beta}_{+}(a^*) =(\tilde{\beta}(a^*)+{\tilde{\beta}}^{\dual}(a^*))/2 = {\beta}^*(a^*)/2$
and
$\tilde{\beta}_{+}(u) =(\tilde{\beta}(u)+{\tilde{\beta}}^{\dual}(u))/2=\beta(u)/2$.

Now suppose that $\beta:(V,\rho)\to
\frak{g}$ is an antisymmetric $\frak{g}$-module homomorphism of
\ewt $\kappa$. Let $b^*\in\frak{g}^*,v\in V$, then
$${\rm
ad}^*_{\ttt{\frak{g}}}(\ttt{\beta}_{+}(a^*+u))(b^*+v) =(1/2)({\rm
ad}^*_{\ttt{\frak{g}}}(\beta^*(a^*))b^*+{\rm
ad}^*_{\ttt{\frak{g}}}(\beta^*(a^*))v+{\rm
ad}^*_{\ttt{\frak{g}}}(\beta(u))b^*+{\rm
ad}^*_{\ttt{\frak{g}}}(\beta(u))v),$$
$${\rm ad}^*_{\ttt{\frak{g}}}(\ttt{\beta}_{+}(b^*+v))(a^*+u) =(1/2)({\rm ad}^*_{\ttt{\frak{g}}}(\beta^*(b^*))a^*+{\rm
ad}^*_{\ttt{\frak{g}}}(\beta^*(b^*))u+{\rm
ad}^*_{\ttt{\frak{g}}}(\beta(v))a^*+{\rm
ad}^*_{\ttt{\frak{g}}}(\beta(v))u).$$

On the other hand, for any $x\in \frak{g},w^*\in V^*$,
{\allowdisplaybreaks {\small
\begin{eqnarray*} &&\langle {\rm
ad}^*_{\ttt{\frak{g}}}(\beta^*(a^*))b^*+{\rm
ad}^*_{\ttt{\frak{g}}}(\beta^*(b^*))a^*,x\rangle=\langle
b^*,[x,\beta^*(a^*)]\rangle+\langle a^*,[x,\beta^*(b^*)]\rangle=0,\\
&&\langle {\rm ad}^*_{\ttt{\frak{g}}}(\beta^*(a^*))b^* +{\rm
ad}^*_{\ttt{\frak{g}}}(\beta^*(b^*))a^*,w^*\rangle=\langle
b^*,[w^*,\beta^*(a^*)]\rangle +\langle
a^*,[w^*,\beta^*(b^*)]\rangle=0,\\
& &\langle {\rm ad}^*_{\ttt{\frak{g}}}(\beta^*(a^*))v+{\rm
ad}^*_{\ttt{\frak{g}}}(\beta(v))a^*,x\rangle=\langle
v,[x,\beta^*(a^*)]\rangle +\langle a^*,[x,\beta(v)]\rangle =-\langle
\beta(\rho(x)v),a^*\rangle+\langle a^*,[x,\beta(v)]\rangle=0,\\
& &\langle {\rm ad}^*_{\ttt{\frak{g}}}(\beta^*(a^*))v+{\rm
ad}^*_{\ttt{\frak{g}}}(\beta(v))a^*,w^*\rangle=\langle
v,[w^*,\beta^*(a^*)]\rangle+\langle a^*,[w^*,\beta(v)]\rangle=0,\\
&&\langle {\rm ad}^*_{\ttt{\frak{g}}}(\beta(u))b^*+{\rm
ad}^*_{\ttt{\frak{g}}}(\beta^*(b^*))u,x\rangle=\langle
b^*,[x,\beta(u)]\rangle+\langle u,[x,\beta^*(b^*)]\rangle=\langle
b^*,[x,\beta(u)]\rangle-\langle \beta(\rho(x)u),b^*\rangle=0,\\
& &\langle {\rm ad}^*_{\ttt{\frak{g}}}(\beta(u))b^*+{\rm
ad}^*_{\ttt{\frak{g}}}(\beta^*(b^*))u,w^*\rangle=\langle
b^*,[w^*,\beta(u)]\rangle+\langle u,[w^*,\beta^*(b^*)]\rangle=0,\\
& &\langle {\rm ad}^*_{\ttt{\frak{g}}}(\beta(u))v+{\rm
ad}^*_{\ttt{\frak{g}}}(\beta(v))u,x\rangle=\langle
v,[x,\beta(u)]\rangle+\langle u,[x,\beta(v)]\rangle=0, \\
& &\langle {\rm ad}^*_{\ttt{\frak{g}}}(\beta(u))v+{\rm
ad}^*_{\ttt{\frak{g}}}(\beta(v))u,w^*\rangle=\langle
v,[w^*,\beta(u)]\rangle+\langle u,[w^*,\beta(v)]\rangle =\langle
\rho(\beta(u))v+\rho(\beta(v))u,w^*\rangle=0.
\end{eqnarray*}}
}

Therefore, ${\rm ad}^*_{\ttt{\frak{g}}}(\ttt{\beta}_{+}
(a^*+u))(b^*+v)+ {\rm ad}^*_{\ttt{\frak{g}}}(\ttt{\beta}_{+}
(b^*+v))(a^*+u)=0.$ Since $\ttt{\ltt{\beta}}_{+}\in
\ttt{\frak{g}}\otimes\ttt{\frak{g}}$ is symmetric, by
Lemma~\mref{le:symmetry}, $\ttt{\beta}_{+}$ is an antisymmetric
$\ttt{\frak g}$-module homomorphism of \ewt
$\kappa$.

Conversely, if $\ttt{\beta}_{+}$ is an antisymmetric
$\ttt{\frak g}$-module homomorphism of \ewt
$\kappa$, then for any $u,v\in V,x\in\frak{g}$,
\begin{eqnarray*}
{\rm ad}^*_{\ttt{\frak{g}}}(\ttt{\beta}_{+}(u))v+{\rm
ad}^*_{\ttt{\frak{g}}}(\ttt{\beta}_{+}(v))u =0&\Leftrightarrow&\rho(\beta(u))v+\rho(\beta(v))u=0,\\
\ttt{\beta}_{+}({\rm
ad}^*_{\ttt{\frak{g}}}(x)v)=[x,\ttt{\beta}_{+}(v)] &\Leftrightarrow&\beta(\rho(x)v)=[x,\beta(v)].
\end{eqnarray*}
 So
$\beta:(V,\rho)\to\frak{g}$ is an antisymmetric $\frak{g}$-module
homomorphism of \ewt $\kappa$.
\end{proof}

\begin{theorem}
Let $\frak{g}$ be a Lie algebra and $(V,\rho)$ be a
$\frak{g}$-module, both with finite $\bfk$-dimensions. Let
$\alpha,\beta:V\to \frakg$ be two linear maps. Using the notations
in Eq.~(\mref{eq:ltl}), $\alpha$ is an extended $\calo$-operator
with extension $\beta$ of \ewt $\kappa$ if and only if
$\ttt{\alpha}_-$ is an extended $\calo$-operator with extension
$\ttt{\beta}_+$ of \ewt $\kappa$. \mlabel{thm:skewgm}
\end{theorem}
\begin{proof}
Note that for any $a^*\in \frak{g}^*,v\in V$, we have
$\tilde{\alpha}_{-}(a^*)=
-\alpha^*(a^*)/2$ and $\tilde{\alpha}_{-}(v)= \alpha(v)/2$
 where
$\alpha^*:\frak{g}^*\rightarrow V^*$ is the dual linear map of
$\alpha$. Suppose that $\alpha$ is an extended $\calo$-operator with
extension $\beta$ of \ewt $\kappa$. Then for any $a^*,b^*\in
\frak{g}^*$, $u,v\in V$, we have

\begin{eqnarray*} & &[\ttt{\alpha}_{-}(u+a^*),
\ttt{\alpha}_{-}(v+b^*)] -\ttt{\alpha}_{-}({\rm
ad}^*_{\ttt{\frak{g}}}(\ttt{\alpha}_{-}(u+a^*))(v+b^*)-{\rm
ad}^*_{\ttt{\frak{g}}}(\ttt{\alpha}_{-}(v+b^*))(u+a^*))\\
&=&(1/4)\{[\alpha(u),\alpha(v)]-[\alpha(u),\alpha^*(b^*)]-[\alpha^*(a^*),\alpha(v)]+[\alpha^*(a^*),\alpha^*(b^*)]\}\\
&& -(1/2) \ttt{\alpha}_{-}({\rm ad}^*_{\ttt{\frak{g}}}(\alpha(u))v+
{\rm ad}^*_{\ttt{\frak{g}}}(\alpha(u))b^*-{\rm
ad}^*_{\ttt{\frak{g}}}(\alpha^*(a^*))v-{\rm
ad}^*_{\ttt{\frak{g}}}(\alpha^*(a^*))b^* \\
&& -{\rm
ad}^*_{\ttt{\frak{g}}}(\alpha(v))u-{\rm
ad}^*_{\ttt{\frak{g}}}(\alpha(v))a^*+{\rm
ad}^*_{\ttt{\frak{g}}}(\alpha^*(b^*))u +{\rm
ad}^*_{\ttt{\frak{g}}}(\alpha^*(b^*))a^*)\\
&=&(1/4)\{[\alpha(u),\alpha(v)]-\alpha(\rho(\alpha(u))v) +\alpha(\rho(\alpha(v))u)-\rho^*(\alpha(u))\alpha^*(b^*)
+\alpha^*({\rm ad}^*(\alpha(u))b^*)\\
&& +
\alpha^*({\rm
ad}^*_{\ttt{\frak{g}}}(\alpha^*(b^*))u)+\rho^*(\alpha(v))\alpha^*(a^*)-\alpha^*({\rm
ad}^*_{\ttt{\frak{g}}}(\alpha^*(a^*))v)-\alpha^*({\rm
ad}^*(\alpha(v))a^*)\}.
\end{eqnarray*}

On the other hand, for any $w\in V$, we have
\begin{eqnarray*}
& &\langle -\rho^*(\alpha(u))\alpha^*(b^*)+\alpha^*({\rm
ad}^*(\alpha(u))b^*)+\alpha^*({\rm
ad}^*_{\ttt{\frak{g}}}(\alpha^*(b^*))u),w\rangle\\
&=&\langle
b^*,\alpha(\rho(\alpha(u))w)+[\alpha(w),\alpha(u)]-\alpha(\rho(\alpha(w))u)\rangle \\
&=&\langle b^*,k[\beta(w),\beta(u)]\rangle\\
&=&\langle b^*,-k\beta(\rho(\beta(u))w)\rangle
\\
&=&\langle k\rho^*(\beta(u))\beta^*(b^*),w\rangle.
\end{eqnarray*}
Thus
$$-\rho^*(\alpha(u))\alpha^*(b^*)+\alpha^*({\rm
ad}^*(\alpha(u))b^*)+\alpha^*({\rm
ad}^*_{\ttt{\frak{g}}}(\alpha^*(b^*))u)=k\rho^*(\beta(u))\beta^*(b^*).$$
Similarly, $$\rho^*(\alpha(v))\alpha^*(a^*)-\alpha^*({\rm
ad}^*_{\ttt{\frak{g}}}(\alpha^*(a^*))v)-\alpha^*({\rm
ad}^*(\alpha(v))a^*)=-k\rho^*(\beta(v))\beta^*(a^*).$$ So
\begin{eqnarray*} & &[\ttt{\alpha}_{-}(u+a^*),
\ttt{\alpha}_{-}(v+b^*)]-\ttt{\alpha}_{-} ({\rm
ad}^*_{\ttt{\frak{g}}}(\ttt{\alpha}_{-} (u+a^*))(v+b^*)-{\rm
ad}^*_{\ttt{\frak{g}}}(\ttt{\alpha}_{-} (v+b^*))(u+a^*))\\
&=&(1/4)\big(\kappa[\beta(u),\beta(v)]+\kappa\rho^*(\beta(u))\beta^*(b^*) -\kappa\rho^*(\beta(v))\beta^*(a^*)\big)\\
&=&(1/4)\big(\kappa[\beta(u),\beta(v)] +\kappa[\beta(u),\beta^*(b^*)]+\kappa[\beta^*(a^*),\beta(v)]\big)\\
&=&\kappa[\ttt{\beta}_{+}(u+a^*), \ttt{\beta}_{+}(v+b^*)].
\end{eqnarray*}

Furthermore, since $\beta$ is an antisymmetric
$\frak{g}$-module homomorphism of \ewt $\kappa$, by
Lemma~\mref{le:syco}, the linear map $\ttt{\beta}_{+}$ from $(\ttt{\frak{g}}^*,{\rm
ad}_{\ttt{\frak{g}}}^*)$ to $\ttt{\frak{g}}$ is an antisymmetric
$\ttt{\frak g}$-module homomorphism of \ewt $\kappa$. Therefore $\ttt{\alpha}_{-}$ is an extended
$\calo$-operator with extension $\ttt{\beta}_{+}$ of \ewt $\kappa$.

Conversely, if $\ttt{\alpha}_{-}$ is an extended $\calo$-operator
with extension $\beta$ of \ewt $\kappa$. Then the linear map $\ttt{\beta}_{+}$ from
$(\ttt{\frak{g}}^*,{\rm ad}_{\ttt{\frak{g}}}^*)$ to $\ttt{\frak{g}}$
is an antisymmetric $\ttt{\frak g}$-module homomorphism of \ewt $\kappa$, which by
Lemma~\mref{le:syco} implies that $\beta$ is an antisymmetric
$\frak{g}$-module homomorphism of \ewt $\kappa$. Moreover, for any
$u,v\in V$, we have that
$$[\ttt{\alpha}_{-}(u), \ttt{\alpha}_{-}(v)] -\ttt{\alpha}_{-} ({\rm
ad}^*_{\ttt{\frak{g}}}(\ttt{\alpha}_{-}(u))v-{\rm
ad}^*_{\ttt{\frak{g}}}(\ttt{\alpha}_{-}(v))u) =\kappa[\ttt{\beta}_{+}(u),
\ttt{\beta}_{+}(v)].$$
Hence
$$[\alpha(u),\alpha(v)]-\alpha(\rho(\alpha(u))v-\rho(\alpha(v))u) =\kappa[\beta(u),\beta(v)].$$
So $\alpha$ is an extended $\calo$-operator with extension $\beta$
of \ewt $\kappa$.
\end{proof}

Theorem~\mref{thm:skewgm} allows us to give the following
characterization of extended $\calo$-operators in terms of solutions
of CYBE in a suitable Lie algebra. In particular, Baxter Lie
algebras are described by CYBE~\mcite{BV1,Bo}.

\begin{coro}
Let $\frak{g}$ be a  Lie algebra and let $(V,\rho)$ be a
$\frak{g}$-module, both with finite $\bfk$-dimension.
\begin{enumerate}
 \item
Let $\alpha, \beta:V\rightarrow \frak{g}$ be linear maps. Then
$\alpha$ is an \tto $\calo$-operator with \bop $\beta$ of \bwt $k$
if and only if $\ttt{\ltt{\alpha}}_-\pm \ttt{\ltt{\beta}}_+$ is a
solution of ECYBE of \bwt $\frac{\kappa+1}{4}$ in
$\frak{g}\ltimes_{\rho^*}V^*$. \mlabel{it:motoaybe6}
\item {\rm (\mcite{Bai})}\; Let
$\alpha:V\rightarrow \frak{g}$ be a linear map.
Then $\alpha$ is an $\calo$-operator of weight zero if and only if
$\ttt{\ltt{\alpha}}_-$ is a skew-symmetric solution of CYBE in $\frak{g}\ltimes_{\rho^*}V^*$. In
particular, a linear map $P:\frak{g}\rightarrow \frak{g}$ is a Rota-Baxter
operator of weight zero if and only if $r=\ttt{\ltt{P}}_-$ is a
skew-symmetric solution of CYBE in
$\frak{g}\ltimes_{{\rm ad}^*}\frak{g}^*$. \mlabel{it:motoaybe5}
\item Let $R:\frak{g}\rightarrow \frak{g}$ be a linear map. Then $(\frak{g},R)$ is a {\bf Baxter Lie algebra}, i.e.,
the following equation holds:
\begin{equation}
[R(x),R(y)]-R([R(x),y]+[x,R(y)])=-[x,y],\quad \forall x,y\in\frak{g}.
\mlabel{eq:-1myb}
\end{equation}
 if and only if $\ttt{\ltt{R}}_-\pm \ttt{\ltt{{\rm
id}}}_+$ is a solution of CYBE in
$\frak{g}\ltimes_{{\rm ad}^*}\frak{g}^*$. \mlabel{it:motoaybe3}
\item Let
$P:\frak{g}\rightarrow \frak{g}$ be a linear map. Then $P$ is a Rota-Baxter operator of
weight $\lambda\neq 0$ if and only if both
$\frac{2}{\lambda}\ttt{\ltt{P}}_-+2\ttt{\ltt{{\rm id}}}$ and
$\frac{2}{\lambda}\ttt{\ltt{P}}_--2\sigma(\ttt{\ltt{{\rm id}}})$ are solutions of
CYBE in $\frak{g}\ltimes_{{\rm ad}^*}\frak{g}^*$.
\mlabel{it:motoaybe4}
\end{enumerate}
\mlabel{co:motoaybe1}
\end{coro}

\begin{proof}
\noindent (\mref{it:motoaybe6}) This follows from
Theorem~\mref{thm:skewgm} and Theorem~\mref{thm:cybea}.
\medskip

\noindent
(\mref{it:motoaybe5}) This follows from Theorem~\mref{thm:skewgm}
for $\kappa =0$ (or $\beta=0$)  and Eq.~(\mref{eq:kuper}).
\medskip

\noindent (\mref{it:motoaybe3}) This follows from Item~(\mref{it:motoaybe6})
in the case that $(V,\rho)=(\frak{g},{\rm ad})$, $\kappa=-1$ and $\beta={\rm id}$.

\medskip

\noindent (\mref{it:motoaybe4}) By~\cite{E}, $P$ is a Rota-Baxter operator of
weight $\lambda\neq 0$ if and only if $\frac{2P}{\lambda}+{\rm
id}$ is an \tto $\calo$-operator with \bop ${\rm id}$ of \bwt $-1$
from $(\frak{g},{\rm ad})$ to $\frak{g}$, i.e., $\frac{2P}{\lambda}+{\rm id}$
satisfies Eq.~(\mref{eq:-1myb}). Then the conclusion follows from
Item~(\mref{it:motoaybe3}).
\end{proof}

\section{\Tto $\calo$-operators and generalized CYBE}
\mlabel{sec:gcyb} In this section, we consider the relationship
between extended $\calo$-operators and the generalized CYBE.

Recall that a
{\bf Lie bialgebra} structure on a Lie algebra $\mathfrak{g}$ is a
skew-symmetric $\bfk$-linear map
$\delta:\mathfrak{g}\rightarrow\mathfrak{g}\otimes\mathfrak{g}$,
called the {\bf cocommutator}, such that $(\mathfrak{g},\delta)$
is a Lie coalgebra and $\delta$ is a $1$-cocycle of $\mathfrak{g}$
with coefficients in $\mathfrak{g}\otimes\mathfrak{g}$, that is,
$\delta$ satisfies the following equation:
$$
\delta([x,y])=({\rm ad}(x)\otimes \id +\id\otimes {\rm
ad}(x))\delta(y)- ({\rm ad}(y)\otimes \id+\id\otimes {\rm
ad}(y))\delta(x),\;\; \forall x,y\in\mathfrak{g}.
$$

\begin{prop}{\rm (\mcite{CP})}
Let $\mathfrak{g}$ be a Lie algebra and
$r\in\mathfrak{g}\otimes\mathfrak{g}$. Define a linear map
$\delta:\frak{g}\to \frak{g}\otimes\frak{g}$
 by
\begin{equation}
\delta(x)=({\rm ad}(x)\otimes \id+\id\otimes {\rm ad}(x))r,\;\;
\forall x\in\mathfrak{g}.\label{eq:1coboundary}
\end{equation}
Then $(\frakg,\delta)$ becomes a Lie coalgebra, i.e.,
$\delta^*:\frakg\otimes\frakg\to\frakg$ defines a Lie algebra
structure on $\mathfrak{g}$, if and only if the following conditions
are satisfied for all $x\in\frak{g}$:
\begin{enumerate}
\item
 $({\rm ad}(x)\otimes \id+\id\otimes{\rm ad}(x))r_+=0$ for $r_+$ defined in Eq.~$($\mref{eq:alphabeta}$)$.
 \mlabel{it:syinvariant}
\item
$({\rm ad}(x)\otimes \id\otimes \id+\id\otimes{\rm ad}(x)\otimes
\id+\id\otimes \id\otimes{\rm
ad}(x))([r_{12},r_{13}]+[r_{12},r_{23}]+[r_{13},r_{23}])=0$.
\end{enumerate}
\mlabel{pp:liebialgebra}
\end{prop}
Such a Lie bialgebra $(\frak{g},\delta)$ is called a {\bf coboundary
Lie bialgebra}~\mcite{CP}.
\begin{defn}{\rm (\mcite{Vai})}
Let $\frak{g}$ be a Lie algebra. The following equation is called
the {\bf generalized classical Yang-Baxter equation (GCYBE)} in
$\frak{g}$:
\begin{equation}
({\rm ad}(x)\otimes \id\otimes \id+\id\otimes{\rm ad}(x)\otimes
\id+\id\otimes \id\otimes{\rm
ad}(x))([r_{12},r_{13}]+[r_{12},r_{23}]+[r_{13},r_{23}])=0,\;
\forall x\in\frak{g}.\mlabel{eq:gcybe1}
\end{equation}
\end{defn}
\begin{lemma}{\rm (\mcite{BGN1})}
Let $\frak{g}$ be a Lie algebra with finite $\bfk$-dimension
and let $r\in \frak{g}\otimes \frak{g}$. Let $[,]_{\delta}$ be the
bracket on $\frakg^*$ induced by Eq.~$($\mref{eq:1coboundary}$)$,
defined by
$$ \langle [a^*,b^*]_\delta, x\rangle = \langle a^*\ot b^*, \delta(x)\rangle, \quad \forall x\in \frakg, a^*, b^*\in \frakg^*.$$
Then for the $\ttl{r}:\frakg^*\to \frakg$ induced from $r$, we have
\begin{equation}
[a^*,b^*]_{\delta}={\rm ad}^*(\ttl{r}(a^*))b^*+{\rm
ad}^*(\ttl{r}^{\dual}(a^*))b^*,\quad \forall
a^*,b^*\in\frak{g}^*.\mlabel{eq:deltap}
\end{equation}
Further, let $\ttl{r}_\pm:\frak g^*\rightarrow \frak g$ be the two
linear maps given by Eq.~(\mref{eq:alphabeta}). If $r_+$ is
invariant, then
\begin{equation}
[a^*,b^*]_{\delta}={\rm ad}^*(\ttl{r}_-(a^*))b^*-{\rm
ad}^*(\ttl{r}_-(a^*))b^*,\quad \forall
a^*,b^*\in\frak{g}^*.\mlabel{eq:braalpha}
\end{equation}
\mlabel{le:bradelta}
\end{lemma}

By Proposition~\mref{pp:liebialgebra} and Lemma~\mref{le:bradelta}, one can get
the following known conclusion:
\begin{coro} {\rm (\mcite{YKS1,YKS2})}
Let $\frak{g}$ be a Lie algebra and $r\in \frak{g}\otimes \frak{g}$.
Suppose that $r$ is skew-symmetric, i.e., $r_+=0$. Then $r$
is a solution of GCYBE if and only if Eq.~$($\mref{eq:braalpha}$)$
defines a Lie bracket on $\frakg^*$.\mlabel{co:syliebra}
\end{coro}

\begin{lemma}
Let $\frak{g}$ be a Lie algebra and $\rho:\frak{g}\to\frak{gl}(V)$
be a representation of $\frakg$ on a vector space $V$. Let
$\alpha:V\to\frak{g}$ be a linear map. Then the bracket
\begin{equation}
[u,v]_{\alpha}:=\rho(\alpha(u))v-\rho(\alpha(v))u,\quad \forall
u,v\in V,\mlabel{eq:doublelie}
\end{equation}
defines a Lie algebra structure on $V$ if and only if the following
equation holds:
\begin{equation}
\rho([\alpha(v),\alpha(u)]-\alpha(\rho(\alpha(v))u-\rho(\alpha(u))v)w+{\rm
cycl.}=0,\quad \forall u,v\in V.\mlabel{eq:bracycl}
\end{equation}
Here for an expression $f(u,v,w)$ in $u,v,w$, $f(u,v,w)+{\rm cycl.}$
means $f(u,v,w)+f(v,w,u)+f(w,u,v)$. \mlabel{le:bracycl}
\end{lemma}
\begin{proof} For any $u,v,w\in V$, we have
\begin{eqnarray*}
{[[u,v]_\alpha,w]}_\alpha&=&\rho(\alpha(\rho(\alpha(u))v-\rho(\alpha(v))u))w-\rho(\alpha(w))\rho(\alpha(u))v+\rho(\alpha(w))\rho(\alpha(v))u,\\
{[[w,u]_\alpha,v]}_\alpha&=& \rho(\alpha(\rho(\alpha(w))u-\rho(\alpha(u))w))v-\rho(\alpha(v))\rho(\alpha(w))u+\rho(\alpha(v))\rho(\alpha(u))w,\\
{[[v,w]_\alpha,u]}_\alpha&=&\rho(\alpha(\rho(\alpha(v))w-\rho(\alpha(w))v))u-\rho(\alpha(u))\rho(\alpha(v))w+\rho(\alpha(u))\rho(\alpha(w))v.
\end{eqnarray*}
Therefore,
\begin{eqnarray*}
& &[[u,v]_\alpha,w]_\alpha+[[w,u]_\alpha,v]_\alpha+[[v,w]_\alpha,u]_\alpha\\
&=&\rho([\alpha(v),\alpha(u)]-\alpha(\rho(\alpha(v))u-\rho(\alpha(u))v))w+{\rm
cycl}, \quad \forall u,v,w\in V.
\end{eqnarray*}
\end{proof}

The following result can be obtained from~\mcite{YKS1, YKS2} in terms of the cocycle conditions. In order to be self-contained,
we give a separate proof.

\begin{theorem} {\rm (\mcite{YKS1,YKS2})}
Let $\frak{g}$ be a Lie algebra with finite $\bfk$-dimension,  $\rho:\frak{g}\to\frak{gl}(V)$ be a finite-dimensional
representation of $\frakg$ and $\alpha:V\to \frakg$ be a linear operator. Using the same notations as in Eq.~$($\mref{eq:ltl}$)$,
$\ttt{\ltt{\alpha}}_{-}\in \ttt{\frak{g}}\otimes\ttt{\frak{g}}$ is a
skew-symmetric solution of GCYBE~(\mref{eq:gcybe1}) if and only if
$\alpha$ satisfies Eq.~$($\mref{eq:bracycl}$)$
 and
 \begin{equation}
[x,B_{\alpha}(u,v)]=B_{\alpha}(\rho(x)u,v)+B_{\alpha}(u,\rho(x)v)\quad
\forall u,v\in V,x\in\frak{g},\mlabel{eq:schoutn}
 \end{equation}
 where
 \begin{equation}
B_{\alpha}(u,v)=[\alpha(u),\alpha(v)]-\alpha(\rho(\alpha(u))v-\rho(\alpha(v))u),\quad
\forall u,v\in V.
\mlabel{eq:cocycle}
 \end{equation}
\mlabel{thm:exogcybe}
\end{theorem}
If a linear operator $\alpha$ satisfies Eq.~(\mref{eq:schoutn}), it is called a {\bf generalized $\calo$-operator}.

One can consider more general Lie brackets and cocycle conditions than given in Eq.~(\mref{eq:doublelie}) and (\mref{eq:cocycle}), by involving (nonzero) Lie structures on the representation spaces and Rota-Baxter operators or, more generally, $\calo$-operators of nonzero weights. We refer the reader to \mcite{BGN1} for some results in this direction. It would be interesting to consider explicit cycle conditions corresponding to Rota-Baxter operators with nonzero weights.

\begin{proof}
By Corollary~\mref{co:syliebra} and Lemma~\mref{le:bracycl} we know
that $\ttt{\ltt{\alpha}}_{-}\in \ttt{\frak{g}}\otimes\ttt{\frak{g}}$
is a skew-symmetric solution of GCYBE~(\mref{eq:gcybe1}) if and only
if for any $u,v,w\in V$ and $a^*,b^*,c^*\in \frak{g}^*$, the following
equation holds {\small
\begin{eqnarray*}
& &{\rm ad}_{\ttt{\frak{g}}}^*([\ttt{\alpha}_{-}(u+a^*),
\ttt{\alpha}_{-}(v+b^*)]- \ttt{\alpha}_{-}({\rm
ad}^*_{\ttt{\frak{g}}} (\ttt{\alpha}_{-}(u+a^*))(v+b^*)-{\rm
ad}^*_{\ttt{\frak{g}}}(\ttt{\alpha}_{-}(v+b^*))(u+a^*)))(w+c^*)\\
& &+{\rm ad}_{\ttt{\frak{g}}}^*([\ttt{\alpha}_{-}(v+b^*),
\ttt{\alpha}_{-}(w+c^*)]-\ttt{\alpha}_{-} ({\rm
ad}^*_{\ttt{\frak{g}}}(\ttt{\alpha}_{-}(v+b^*))(w+c^*)-{\rm
ad}^*_{\ttt{\frak{g}}}(\ttt{\alpha}_{-}(w+c^*))(v+b^*)))(u+a^*)\\
& & +{\rm ad}_{\ttt{\frak{g}}}^*([\ttt{\alpha}_{-}(w+c^*),
\ttt{\alpha}_{-}(u+a^*)]-\ttt{\alpha}_{-} ({\rm
ad}^*_{\ttt{\frak{g}}}(\ttt{\alpha}_{-}(w+c^*))(u+a^*)-{\rm
ad}^*_{\ttt{\frak{g}}}(\ttt{\alpha}_{-}(u+a^*))(w+c^*)))(v+b^*)=0.
\end{eqnarray*}
} By the proof of Theorem~\mref{thm:skewgm}, this is equivalent to
\begin{eqnarray*}
& &{\rm
ad}_{\ttt{\frak{g}}}^*([\alpha(u),\alpha(v)]-\alpha(\rho(\alpha(u))v)+
\alpha(\rho(\alpha(v))u)-\rho^*(\alpha(u))\alpha^*(b^*)
+\alpha^*({\rm ad}^*(\alpha(u))b^*)\\
& &+\alpha^*({\rm
ad}^*_{\ttt{\frak{g}}}(\alpha^*(b^*))u)+\rho^*(\alpha(v))\alpha^*(a^*)-\alpha^*({\rm
ad}^*_{\ttt{\frak{g}}}(\alpha^*(a^*))v)-\alpha^*({\rm
ad}^*(\alpha(v))a^*))(w+c^*)\\
& &+{\rm
ad}_{\ttt{\frak{g}}}^*([\alpha(v),\alpha(w)]-\alpha(\rho(\alpha(v))w)+
\alpha(\rho(\alpha(w))v)-\rho^*(\alpha(v))\alpha^*(c^*)
+\alpha^*({\rm ad}^*(\alpha(v))c^*)\\
& &+\alpha^*({\rm
ad}^*_{\ttt{\frak{g}}}(\alpha^*(c^*))v)+\rho^*(\alpha(w))\alpha^*(b^*)-\alpha^*({\rm
ad}^*_{\ttt{\frak{g}}}(\alpha^*(b^*))w)-\alpha^*({\rm
ad}^*(\alpha(w))b^*))(u+a^*)\\
& &+{\rm
ad}_{\ttt{\frak{g}}}^*([\alpha(w),\alpha(u)]-\alpha(\rho(\alpha(w))u)+
\alpha(\rho(\alpha(u))w)-\rho^*(\alpha(w))\alpha^*(a^*)
+\alpha^*({\rm ad}^*(\alpha(w))a^*)\\
& &+\alpha^*({\rm
ad}^*_{\ttt{\frak{g}}}(\alpha^*(a^*))w)+\rho^*(\alpha(u))\alpha^*(c^*)-\alpha^*({\rm
ad}^*_{\ttt{\frak{g}}}(\alpha^*(c^*))u)-\alpha^*({\rm
ad}^*(\alpha(u))c^*))(v+b^*)=0.
\end{eqnarray*}
Since for any $s^*\in V^*$ and $a^*\in\frak{g}^*$, we have ${\rm
ad}_{\ttt{\frak{g}}}^*(s^*)a^*=0$, the above equation is equivalent
to the following equations:
{\allowdisplaybreaks
\begin{eqnarray}
& &{\rm
ad}_{\ttt{\frak{g}}}^*([\alpha(u),\alpha(v)]-\alpha(\rho(\alpha(u))v-\rho(\alpha(v))u))w+{\rm
cycl}.=0,\label{eq:equiva1}\\
& &{\rm
ad}_{\ttt{\frak{g}}}^*(-\rho^*(\alpha(u))\alpha^*(b^*)+\alpha^*({\rm
ad}^*(\alpha(u))b^*)+\alpha^*({\rm
ad}_{\ttt{\frak{g}}}^*(\alpha^*(b^*))u))w+{\rm
ad}_{\ttt{\frak{g}}}^*(\rho^*(\alpha(w))\alpha^*(b^*)\label{eq:equiva2}\\
& &-\alpha^*({\rm
ad}_{\ttt{\frak{g}}}^*(\alpha^*(b^*))w)-\alpha^*({\rm
ad}(\alpha(w))b^*))u+{\rm
ad}_{\ttt{\frak{g}}}^*([\alpha(w),\alpha(u)] -\alpha(\rho(\alpha(w))u-\rho(\alpha(u))w))b^*=0,\notag\\
& &{\rm
ad}_{\ttt{\frak{g}}}^*(\rho^*(\alpha(v))\alpha^*(a^*)-\alpha^*({\rm
ad}_{\ttt{\frak{g}}}^*(\alpha^*(a^*))v)-\alpha^*({\rm
ad}^*(\alpha(v))a^*))w+{\rm
ad}_{\ttt{\frak{g}}}^*([\alpha(v),\alpha(w)]\notag\\
& &-\alpha(\rho(\alpha(v))w)+\alpha(\rho(\alpha(w))v))a^*+{\rm
ad}_{\ttt{\frak{g}}}^*(-\rho^*(\alpha(w))\alpha^*(a^*)+\alpha^*({\rm
ad}^*(\alpha(w))a^*)\label{eq:equiva3}\\
& &+\alpha^*({\rm
ad}_{\ttt{\frak{g}}}^*(\alpha^*(a^*))w))v=0,\notag\\
& &{\rm ad}_{\ttt{\frak{g}}}^*([\alpha(u),\alpha(v)]
-\alpha(\rho(\alpha(u))v)+\alpha(\rho(\alpha(v))u))c^*+{\rm
ad}_{\ttt{\frak{g}}}^*(-\rho^*(\alpha(v))\alpha^*(c^*)+\alpha^*({\rm
ad}^*(\alpha(v))c^*)\label{eq:equiva4}\\
& &+\alpha^*({\rm ad}_{\ttt{\frak{g}}}^*(\alpha^*(c^*))v))u+{\rm
ad}_{\ttt{\frak{g}}}^*(\rho^*(\alpha(u))\alpha^*(c^*)-\alpha^*({\rm
ad}_{\ttt{\frak{g}}}^*(\alpha^*(c^*))u)-\alpha^*({\rm
ad}^*(\alpha(u))c^*))v=0.\notag
\end{eqnarray}
}
We shall prove
\begin{enumerate}
\item
Eq.~(\mref{eq:equiva1}) $\Leftrightarrow $ Eq.~(\mref{eq:bracycl}),
\item
Eq.~(\mref{eq:equiva2}) $\Leftrightarrow$ Eq.~(\mref{eq:equiva3})
$\Leftrightarrow$ Eq.~(\mref{eq:equiva4}) $\Leftrightarrow$
Eq.~(\mref{eq:schoutn}).
\end{enumerate}
The proofs of these statements are similar. So we
just prove that Eq.~(\mref{eq:equiva2}) holds if and only if
Eq.~(\mref{eq:schoutn}) holds. Let $LHS$ denotes the left-hand
side of Eq.~(\mref{eq:equiva2}). For any $v^*,s^*\in V^*$, $w\in
V$ and $x,y\in\frakg$, we have $\langle {\rm
ad}^*_{\tilde{\frakg}}(v^*)w,s^*\rangle=0, \langle{\rm
ad}^*_{\tilde{\frakg}}([x,y])b^*,s^*\rangle=0$. Thus we obtain
$\langle LHS,s^*\rangle=0$.

Further, for any $x\in\frak{g}$,
\begin{eqnarray*}
\langle LHS,x\rangle&=&\langle
w,[\rho^*(\alpha(u))\alpha^*(b^*)-\alpha^*({\rm
ad}^*(\alpha(u))b^*)-\alpha^*({\rm
ad}_{\ttt{\frak{g}}}^*(\alpha^*(b^*))u),x]\rangle\\
& &+\langle u,[-\rho^*(\alpha(w))\alpha^*(b^*)+\alpha^*({\rm
ad}_{\ttt{\frak{g}}}^*(\alpha^*(b^*))w)+\alpha^*({\rm
ad}^*(\alpha(w))b^*),x]\rangle \\
& &+\langle
b^*,[x,[\alpha(w),\alpha(u)]-\alpha(\rho(\alpha(w))u-\rho(\alpha(u))w)]\rangle\\
&=&\langle
-\alpha(\rho(\alpha(u))\rho(x)w)+[\alpha(u),\alpha(\rho(x)w)],b^*\rangle-
\langle[\alpha(\rho(x)w),\alpha^*(b^*)],u\rangle\\
& &+\langle\alpha(\rho(\alpha(w))\rho(x)u)+
[\alpha(\rho(x)u),\alpha(w)],b^*\rangle+ \langle
[\alpha(\rho(x)u),\alpha^*(b^*)],w\rangle\\
& & + \langle
b^*,[x,[\alpha(w),\alpha(u)]-\alpha(\rho(\alpha(w))u-\rho(\alpha(u))w)]\rangle\\
&=&\langle[\alpha(u),\alpha(\rho(x)w)]-\alpha(\rho(\alpha(u))\rho(x)w)+
\alpha(\rho(\alpha(\rho(x)w))u),b^*\rangle\\
& &+\langle[\alpha(\rho(x)u),\alpha(w)]+
\alpha(\rho(\alpha(w))\rho(x)u)-\alpha(\rho(\alpha(\rho(x)u))w),b^*\rangle\\
& &
+\langle[x,[\alpha(w),\alpha(u)]-\alpha(\rho(\alpha(w))u-\rho(\alpha(u))w)],b^*\rangle.
\end{eqnarray*}
So Eq.~(\mref{eq:equiva2}) holds if and only if
Eq.~(\mref{eq:schoutn}) holds.
\end{proof}

\begin{remark}{\rm With the notations as above, it is in fact straightforward to show that the bracket
\begin{equation}
[u,v]_{\alpha}:=\rho(\alpha(u))v-\rho(\alpha(v))u,\quad \forall
u,v\in V,
\end{equation}
defines a Lie algebra structure on $V$ if and only if Eq.~(\ref{eq:bracycl}) hold.}\end{remark}

\begin{coro}
Let $\frak{g}$ be a Lie algebra with finite $\bfk$-dimension.
\begin{enumerate}
\item Let $(\frak{k},\pi)$ be a $\frak{g}$-Lie algebra with finite
$\bfk$-dimension. Let $\alpha,\beta:\frak{k}\to\frakg$ be two
linear maps such that $\alpha$ is an extended $\calo$-operator of
weight $\lambda$ with extension $\beta$ of \ewt $(\kappa,\mu)$.
Then
$\ttt{\ltt{\alpha}}_{-}\in(\frakg\ltimes_{\pi^*}\frak{k}^*)\otimes(\frakg\ltimes_{\pi^*}\frak{k}^*)$
is a skew-symmetric solution of GCYBE if and only if the following
equations hold:
\begin{equation}
\lambda\pi(\alpha([u,v]_{\frak{k}}))w+\lambda\pi(\alpha([w,u]_{\frak{k}}))v+
\lambda\pi(\alpha([v,w]_{\frak{k}}))u=0,\mlabel{eq:lakm1}
\end{equation}
\begin{equation}
\lambda[x,\alpha([u,v]_{\frak{k}})]_{\frakg}=\lambda\alpha([\pi(x)u,v]_{\frak{k}})+
\lambda\alpha([u,\pi(x)v]_{\frak{k}}), \quad \forall x\in\frakg,u,v,w\in\frak{k}.
\mlabel{eq:lakm2}
\end{equation}
In particular, if
$\lambda=0$, i.e., $\alpha$ is an extended
$\calo$-operator of weight 0 with extension $\beta$ of
\ewt $(\kappa,\mu)$,
then $\ttt{\ltt{\alpha}}_{-}\in(\frakg\ltimes_{\pi^*}\frak{k}^*)\otimes(\frakg\ltimes_{\pi^*}\frak{k}^*)$
is a skew-symmetric solution of GCYBE. \mlabel{it:lakm1} \item Let
$(\frak{k},\pi)$ be a $\frak{g}$-Lie algebra with finite
$\bfk$-dimension. Let $\alpha:\frak{k}\to\frakg$ an
$\calo$-operator of weight $\lambda$. Then
$\ttt{\ltt{\alpha}}_{-}\in(\frakg\ltimes_{\pi^*}\frak{k}^*)\otimes(\frakg\ltimes_{\pi^*}\frak{k}^*)$
is a skew-symmetric solution of GCYBE if and only if
Eq.~$($\mref{eq:lakm1}$)$ and Eq.~$($\mref{eq:lakm2}$)$
hold.\mlabel{it:lakm2} \item Let $\rho:\frak{g}\to\frak{gl}(V)$ be
a finite dimensional representation of $\frakg$. Let
$\alpha,\beta:\frak{k}\to\frakg$ be two linear maps such that
$\alpha$ is an extended $\calo$-operator with extension $\beta$ of
\ewt $\kappa$. Then
$\ttt{\ltt{\alpha}}_{-}\in(\frakg\ltimes_{\rho^*}V^*)\otimes(\frakg\ltimes_{\rho^*}V^*)$
is a skew-symmetric solution of GCYBE.\mlabel{it:lakm3}
\end{enumerate}
\end{coro}
\begin{proof}
(\mref{it:lakm1}) Since $\alpha$ is an extended $\calo$-operator of
weight $\lambda$ with extension $\beta$ of \ewt $(\kappa,\mu)$, for
any $u,v\in\frak{k}$, we have
$$B_{\alpha}(u,v)=[\alpha(u),\alpha(v)]_{\frakg}-\alpha(\pi(\alpha(u))v-\pi(\alpha(v))u)=
\lambda\alpha([u,v]_{\frak{k}})+\kappa[\beta(u),\beta(v)]_{\frakg}+\mu\beta([u,v]_{\frak{k}}).$$
Thus, by Theorem~\mref{thm:exogcybe},
$\ttt{\ltt{\alpha}}_{-}\in(\frakg\ltimes_{\pi^*}\frak{k}^*)\otimes(\frakg\ltimes_{\pi^*}\frak{k}^*)$
is a skew-symmetric solution of GCYBE if and only if the following
equations hold:
\begin{equation}
\pi(\lambda\alpha([u,v]_{\frak{k}})+\kappa[\beta(u),\beta(v)]_{\frakg}+\mu\beta([u,v]_{\frak{k}}))w+{\rm
cycl}.=0,\mlabel{eq:prcogyb1}
\end{equation}
\begin{eqnarray}
&&[x,\lambda\alpha([u,v]_{\frak{k}})+\kappa[\beta(u),\beta(v)]_{\frakg}+\mu\beta([u,v]_{\frak{k}})]_{\frakg}
\mlabel{eq:prcogyb2}\\
&=&\lambda\alpha([\pi(x)u,v]_{\frak{k}})+\kappa[\beta(\pi(x)u),\beta(v)]_{\frakg}+
\mu\beta([\pi(x)u,v]_{\frak{k}})\notag\\
& &+\lambda\alpha([u,\pi(x)v]_{\frak{k}})+
\kappa[\beta(u),\beta(\pi(x)v)]_{\frakg} +\mu\beta([u,\pi(x)v]_{\frak{k}}),\quad \forall x\in\frakg,u,v\in\frak{k}.
\notag
\end{eqnarray}
On the other hand, for any
$u,v,w\in\frak{k}$, we have
\begin{eqnarray*}
\kappa\pi([\beta(w),\beta(u)]_{\frakg})v+\kappa\pi([\beta(v),\beta(w)]_{\frakg})u&=&
\kappa\pi(\beta(\pi(\beta(w))u))v+\kappa\pi(\beta(\pi(\beta(v))w))u\\
&=&-\kappa\pi(\beta(v))\pi(\beta(w))u-\kappa\pi(\beta(u))\pi(\beta(v))w\\
&=&\kappa\pi(\beta(v))\pi(\beta(u))w-\kappa\pi(\beta(u))\pi(\beta(v))w\\
&=&\kappa\pi([\beta(v),\beta(u)]_{\frakg})w.
\end{eqnarray*}
Therefore, $\kappa\pi([\beta(u),\beta(v)]_{\frakg})w+{\rm cycl}.=0$.
Moreover,
\begin{eqnarray*}
\mu\pi(\beta([u,v]_{\frak{k}}))w&=&-\mu\pi(\beta(w))[u,v]_{\frak{k}}\\
&=&-\mu[\pi(\beta(w))u,v]_{\frak{k}}-\mu[x,\pi(\beta(w))v]_{\frak{k}}\\
&=&-\mu\pi(\beta([w,u]_{\frak{k}}))v+\mu[x,\pi(\beta(v))w]_{\frak{k}}\\
&=&-\mu\pi(\beta([w,u]_{\frak{k}}))v-\mu\pi(\beta([v,w]_{\frak{k}}))u.
\end{eqnarray*}
Therefore
$\mu\pi(\beta([u,v]_{\frak{k}}))w+\mu\pi(\beta([w,u]_{\frak{k}}))v+\mu\pi(\beta([v,w]_{\frak{k}}))u=0$.
So Eq.~(\mref{eq:lakm1}) holds if and only if
Eq.~(\mref{eq:prcogyb1}) holds. Furthermore, for any $x\in\frakg$,
we have that
$$[x,\kappa[\beta(u),\beta(v)]_{\frakg}]_{\frakg}=\kappa[[x,\beta(u)]_{\frakg},\beta(v)]_{\frakg}+
\kappa[\beta(u),[x,\beta(v)]_{\frakg}]_{\frakg}=\kappa[\beta(\pi(x)u),\beta(v)]_{\frakg}+
\kappa[\beta(u),\beta(\pi(x)v)]_{\frakg}$$ and
$$[x,\mu\beta([u,v]_{\frak{k}})]_{\frakg}=\mu\beta(\pi(x)[u,v]_{\frak{k}})=\mu\beta([\pi(x)u,v]_{\frak{k}})+
\mu\beta([u,\pi(x)v]_{\frak{k}}).$$ Therefore,
Eq.~(\mref{eq:lakm2}) holds if and only if
Eq.~(\mref{eq:prcogyb2}) holds. In conclusion,
$\ttt{\ltt{\alpha}}_{-}\in(\frakg\ltimes_{\pi^*}\frak{k}^*)\otimes(\frakg\ltimes_{\pi^*}\frak{k}^*)$
is a skew-symmetric solution of GCYBE if and only if
Eq.~(\mref{eq:lakm1}) and Eq.~(\mref{eq:lakm2}) hold.
\medskip

\noindent (\mref{it:lakm2}) This follows from Item~(\mref{it:lakm1})
by setting $\kappa =\mu=0$.
\medskip

\noindent (\mref{it:lakm3}) This follows from Item~(\mref{it:lakm1})
by setting $(\frak{k},\pi)=(V,\rho)$.
\end{proof}

\section{Rota-Baxter operators, $\calo$-operators and relative differential operators}
\mlabel{sec:rbod}
In this section, we show that an $\calo$-operator can be recovered
from a Rota-Baxter operator on a large space. We also introduce a differential variation of the $\calo$-operator and study its relation with $\calo$-operators.

\subsection{Rota-Baxter operators and $\calo$-operators}
We start with the relationship between $\calo$-operators on a $\frakg$-Lie algebra and Rota-Baxter operators.
\begin{prop}
Let $\frak{g}$ be a Lie algebra and $(\frak{k},\pi)$ be a
$\frak{g}$-Lie algebra. Let $\alpha:\frak{k}\to\frak{g}$ be a linear
map and let $\lambda\in \bfk$. Then the following statements are equivalent.
\begin{enumerate}
\item
The linear map $\alpha$ is an $\calo$-operator of weight
$\lambda$.\mlabel{it:crbb1}
\item
The linear map
\begin{equation}
\otr{\alpha}:\frak{g}\mltimes_{\pi}\frak{k}\to\frak{g}\mltimes_{\pi}\frak{k},\quad
\otr{\alpha}(x,u)=(\alpha(u)-\lambda x,0),\quad \forall
x\in\frak{g},u\in\frak{k},
\end{equation}
is a Rota-Baxter operator of weight $\lambda$.\mlabel{it:crbb2}
\item
The linear map
\begin{equation}
-\lambda \id-\otr{\alpha}: \frak{g}\mltimes_{\pi}\frak{k}\to\frak{g}\mltimes_{\pi}\frak{k},\quad
(-\lambda \id -\otr{\alpha})(x,u)=(-\alpha(u),-\lambda u),\quad \forall
x\in\frak{g},u\in\frak{k},
\end{equation}
is a Rota-Baxter operator of weight $\lambda$.\mlabel{it:crbb3}
\end{enumerate}
\mlabel{pp:crbb}
\end{prop}
\begin{proof}
(\mref{it:crbb1})$\Leftrightarrow$(\mref{it:crbb2}). Let
$x,y\in\frak{g},u,v\in\frak{k}$. Then we have
\begin{eqnarray*}
[\otr{\alpha}(x,u),\otr{\alpha}(y,v)]&=&(\lambda^2[x,y]-\lambda[x,\alpha(v)]-\lambda[\alpha(u),y]+
[\alpha(u),\alpha(v)],0),\\
\otr{\alpha}([\otr{\alpha}(x,u),(y,v)])&=&(\lambda^2[x,y]-\lambda[\alpha(u),y]-
\lambda\alpha(\pi(x)v)+\alpha(\pi(\alpha(u))v),0),\\
\otr{\alpha}([(x,u),\otr{\alpha}(y,v)])&=&(\lambda^2[x,y]-\lambda[x,\alpha(v)]+\lambda\alpha(\pi(y)u)-
\alpha(\pi(\alpha(v))u),0),\\
\lambda\otr{\alpha}([(x,u),(y,v)])&=&(-\lambda^2[x,y]+\lambda\alpha(\pi(x)v)-\lambda\alpha(\pi(y)u)+
\lambda\alpha([u,v]),0).
\end{eqnarray*}
Therefore
$$[\alpha(u),\alpha(v)]=\alpha(\pi(\alpha(u))v-\pi(\alpha(v))u)+\lambda\alpha([u,v])$$
if and only if
$$[\otr{\alpha}(x,u),\otr{\alpha}(y,v)]=\otr{\alpha}([\otr{\alpha}(x,u),(y,v)])+
\otr{\alpha}([(x,u),\otr{\alpha}(y,v)])+\lambda\otr{\alpha}([(x,u),(y,v)]),$$
for any $x,y\in\frak{g},u,v\in\frak{k}$.

(\mref{it:crbb2})$\Leftrightarrow$(\mref{it:crbb3}). This follows
from the following basic fact on Rota-Baxter operators: a
linear map $P$ on a Lie algebra is a Rota-Baxter operator of weight
$\lambda\in\bfk$ if and only if $-\lambda{\rm id}-P$ is a
Rota-Baxter operator of weight $\lambda\in\bfk$.
\end{proof}

For $\calo$-operators on a $\frakg$-module, we have
\begin{coro}
Let $\frak{g}$ be a Lie algebra and $(V,\rho)$ be a $\frakg$-module.
Let $\alpha:V\to\frak{g}$ be a linear map. Let $\lambda$ and $\mu\neq 0$ be in $\bfk$. The following
statements are equivalent.
\begin{enumerate}
\item The linear map $\alpha$ is an $\calo$-operator (of weight 0).
\item The linear map
\begin{equation}
\otr{\alpha}:\frak{g}\ltimes_{\rho} V\to\frak{g}\ltimes_{\rho}
V,\quad \otr\alpha(x,u)=(\mu\alpha(u)-\lambda x,0),\quad
\forall x\in\frak{g},u\in V,
\end{equation}
 is a Rota-Baxter operator of weight $\lambda$.
\item
The linear map
\begin{equation}
-\lambda \id - \otr{\alpha}:\frak{g}\ltimes_{\rho}
V\to\frak{g}\ltimes_{\rho} V,\quad (-\lambda\id
-\otr{\alpha})(x,u)=(-\mu\alpha(u),-\lambda u),\quad \forall
x\in\frak{g},u\in V,
\end{equation}
is a Rota-Baxter operator of weight $\lambda$.
\end{enumerate}
\mlabel{co:crbb}
\end{coro}
\begin{proof}
Since a $\frakg$-module $(V,\rho)$ is a $\frakg$-Lie algebra when $V$ is equipped with the zero bracket, the corollary follows from Proposition~\mref{pp:crbb} and the simple fact that a linear operator $\alpha:V\to \frakg$ is an $\calo$-operator if and only if $\mu\alpha$ is one for $0\neq \mu\in \bfk$.
\end{proof}

By a similar argument as that for Proposition~\mref{pp:crbb}, we also obtain the following relation of invertible $\calo$-operators with Rota-Baxter operators.
\begin{prop}
Let $\frak{g}$ be a Lie algebra and $(V,\rho)$ be a $\frakg$-module.
Let $\alpha:V\to\frak{g}$ be an invertible linear map. Let $\lambda$, $\mu_1\neq 0$ and $\mu_2\neq \pm\lambda$ be in $\bfk$. Then $\alpha$
is an $\calo$-operator of weight 0 if and only if
\begin{equation}
\otr{\alpha}(x,u)=\left(\mu_1\alpha(u)-\frac{\mu_2+\lambda}{2}x,\frac{\lambda^2-\mu_2^2}{4\mu_1}\alpha^{-1}(x)+
\frac{\mu_2-\lambda}{2}u\right), \quad \forall x\in \frakg, u\in V, \end{equation}
is a Rota-Baxter operator of weight $\lambda$ on
$\frak{g}\ltimes_{\rho}V$.
\end{prop}

\subsection{Rota-Baxter operators and relative differential operators}
We first define a relative version of the differential operator which can also be regarded as a differential variation of the $\calo$-operator.
\begin{defn}
{\rm Let $\frak{g}$ be a Lie algebra.
\begin{enumerate}
\item Let $(\frak{k},\pi)$ be a
$\frak{g}$-Lie algebra. A linear map $f:\frak{g}\to\frak{k}$ is
called a {\bf relative differential operator (on $\frak k$) of weight $\lambda$} if
\begin{equation}
f([x,y]_{\frak{g}})=\pi(x)f(y)-\pi(y)f(x)+\lambda
[f(x),f(y)]_{\frak{k}},\quad \forall x,y\in\frak{g}.
\end{equation}
\item
Let $(V,\rho)$ be a $\frakg$-module. A linear map $f:\frakg \to V$
is called a {\bf relative differential operator (on $V$)} if
\begin{equation}
f([x,y]_{\frak{g}})=\pi(x)f(y)-\pi(y)f(x),\quad \forall x,y\in\frak{g}.
\end{equation}
\end{enumerate}
}
\end{defn}

A relative differential operator on a $\frakg$-module $V$ can be regarded as a special case of a relative differential operator on a $\frakg$-Lie algebra $\frak k$ when $V$ is equipped with the trivial Lie bracket.

\begin{prop}
Let $\frak{g}$ be a Lie algebra and $(\frak{k},\pi)$ be a
$\frak{g}$-Lie algebra. Let $f:\frak{g}\to\frak{k}$ be a linear map.
Then the following statement are equivalent.
\begin{enumerate}
\item
The linear map $f$ is a relative differential operator of weight
1.\mlabel{it:difro1}
\item
The linear map
\begin{equation}
\otr{f}:\frak{g}\mltimes_{\pi}\frak{k}\to\frak{g}\mltimes_{\pi}\frak{k},\quad
\otr{f}(x,u)=(x,f(x)),\quad \forall x\in\frak{g},u\in\frak{k},
\end{equation}
is a Rota-Baxter operator of weight $-1$. \mlabel{it:difro2}
\item
The
linear map
\begin{equation}
\id-\otr{f}:\frak{g}\mltimes_{\pi}\frak{k}\to\frak{g}\mltimes_{\pi}\frak{k},\quad
(\id-\otr{f})(x,u)=(0,u-f(x)),\quad \forall
x\in\frak{g},u\in\frak{k},
\end{equation}
is a Rota-Baxter operator of weight $-1$. \mlabel{it:difro3}
\end{enumerate}
\mlabel{pp:crdiff}
\end{prop}
\begin{proof}
(\mref{it:difro1})$\Leftrightarrow$(\mref{it:difro2}). Let
$x,y\in\frak{g},u,v\in\frak{k}$. Then we have
\begin{eqnarray*}
[\otr{f}(x,u),\otr{f}(y,v)]&=&([x,y],\pi(x)f(y)-\pi(y)f(x)+[f(x),f(y)]);\\
\otr{f}([\otr{f}(x,u),(y,v)])&=&([x,y],f([x,y]));\\
\otr{f}([(x,u),\otr{f}(y,v)])&=&([x,y],f([x,y]));\\
\otr{f}([(x,u),(y,v)])&=&([x,y],f([x,y])).
\end{eqnarray*}
Therefore
$$[\otr{f}(x,u),\otr{f}(y,v)]=\otr{f}([\otr{f}(x,u),(y,v)])+\otr{f}([(x,u),\otr{f}(y,v)])-
\otr{f}([(x,u),(y,v)])$$ if and only if
$$f([x,y])=\pi(x)f(y)-\pi(y)f(x)+[f(x),f(y)].$$
for any $x,y\in\frak{g},u,v\in\frak{k}$.

(\mref{it:difro2})$\Leftrightarrow$(\mref{it:difro3}). This follows
from the same reason as in the case of Proposition~\mref{pp:crbb}.
\end{proof}

By the same argument as for  Corollary~\mref{co:crbb}, we have:

\begin{coro}
Let $\frak{g}$ be a Lie algebra and $(V,\rho)$ be a $\frakg$-module.
Let $f:\frak{g}\to V$ be a linear map. Let $\lambda, \mu\in \bfk$ be nonzero. Then the following
statements are equivalent.
\begin{enumerate}
\item The linear map $f$ is a relative differential operator.
\item The linear map
\begin{equation}
\otr{f}:\frak{g}\ltimes_{\rho} V\to\frak{g}\ltimes_{\rho} V,\quad
\otr{f}(x,u)=(-\lambda x,\mu f(x)),\quad \forall x\in\frak{g},u\in
V,
\end{equation}
is a Rota-Baxter operator of weight $\lambda$.
\item
The linear map
\begin{equation}
-\lambda\id-\otr{f}:\frak{g}\ltimes_{\rho}
V\to\frak{g}\ltimes_{\rho} V,\quad
(-\lambda\id-\otr{f})(x,u)=(0,-\mu f(x)-\lambda u),\quad \forall
x\in\frak{g},u\in V,
\end{equation}
is a Rota-Baxter operator of weight $\lambda$.
\end{enumerate}
\end{coro}

\end{document}